\documentclass[pra,notitlepage,twocolumn,superscriptaddress,nofootinbib,longbibliography]{revtex4-2}

\usepackage{enumerate}
\usepackage{amsfonts,amssymb,amsmath}
\usepackage[]{graphics,graphicx,epsfig}
\usepackage{amsthm}
\usepackage{amssymb}
\usepackage{mathbbol}
\usepackage{graphicx}
\usepackage{dcolumn}
\usepackage{natbib}
\usepackage{color}
\usepackage{multirow}
\usepackage{ulem}
\usepackage{bm}
\usepackage{url}
\usepackage{float}
\usepackage{enumitem}
\usepackage[usenames,dvipsnames]{xcolor}
\usepackage{thmtools}
\usepackage{nicematrix}
\usepackage{svg}
\usepackage{stmaryrd}

\usepackage{tikz}
\usepackage{environ}
\usepackage{tikz-cd}
\tikzcdset{diagrams={nodes={inner sep=2pt}}}

\usepackage[colorlinks = true, citecolor=blue, linkcolor=blue, urlcolor=blue]{hyperref}
\newcommand*{\fullref}[1]{\hyperref[{#1}]{\autoref*{#1} \nameref*{#1}}}

\newcommand{\C}{\mathbb{C}}

\newcommand{\ketbra}[1]{|{#1}\>\mkern-5mu\<{#1}|}
\renewcommand{\>}{\rangle}
\newcommand{\<}{\langle}

\renewcommand{\mod}[1]{\;\mathrm{mod}\,{#1}}

\newcommand{\vcirc}{{\text{\scalebox{1.7}{$\bullet$}}}}
\newcommand{\vsq}{{\mathbin{\blacksquare}}}

\newtheorem{theorem}{Theorem}

\newtheorem{proposition}[theorem]{Proposition}

\theoremstyle{definition}
\newtheorem{definition}{Definition}

\newcommand{\CQT}{Centre for Quantum Technologies, National University of Singapore, 3 Science Drive 2, Singapore 117543.\looseness=-1}
\newcommand{\NTU}{Nanyang Quantum Hub, School of Physical and Mathematical Sciences, Nanyang Technological University, Singapore 639673.\looseness=-1}
\newcommand{\IHPC}{A*STAR Quantum Innovation Centre (Q.InC), Institute of High Performance Computing (IHPC), Agency for Science, Technology and Research (A*STAR), 1 Fusionopolis Way, \#16-16 Connexis, Singapore, 138632, Republic of Singapore.\looseness=-1}
\newcommand{\CQuERE}{Centre for Quantum Engineering, Research and Education, TCG CREST, Sector V, Salt Lake, Kolkata 700091, India.\looseness=-1}
\newcommand{\sutd}{Science, Mathematics and Technology Cluster, Singapore University of Technology and Design, 8 Somapah Road, Singapore 487372, Singapore}

\begin{document}

\title{Simple Construction of Qudit Floquet Codes on a Family of Lattices}

\author{Andrew Tanggara}
\email{andrew.tanggara@gmail.com}
\affiliation{\CQT}
\affiliation{\NTU}

\author{Mile Gu}
\email{mgu@quantumcomplexity.org}
\affiliation{\NTU}
\affiliation{\CQT}

\author{Kishor Bharti}
\email{kishor.bharti1@gmail.com}
\affiliation{\IHPC}
\affiliation{\CQuERE}
\affiliation{\sutd}

\date{\today}

\begin{abstract}
    Dynamical quantum error-correcting codes (QECC) offer wider possibilities in how one can protect logical quantum information from noise and perform fault-tolerant quantum computation compared to static QECCs.
    A family of dynamical QECCs called the ``Floquet codes'' consists of a periodic sequence of two-body measurements that enables error-correction on many-body systems, relaxing hardware implementation requirements and improving error-correction reliability.
    Existing results on Floquet codes has been focused on qubits, two-level quantum systems, with very little attention given on higher dimensional quantum systems, or qudits.
    We bridge this gap by proposing a simple, yet general construction of qudit Floquet codes based on a simple set of conditions on the sequence two-body measurements defining the code.
    Moreover, this construction applies to a large family of configurations of qudits on the vertices of a three-colorable lattice which connectivity represented by the edges.
    We show that this construction includes the existing constructions of both qubit and qudit Floquet codes as special cases.
    In addition, any qudit Floquet code obtained by our construction achieves a rate of encoded logical qudits over physical qudits approaching $\frac{1}{2}$ as the number of physical qudits in total and on the faces of the lattice grows larger, as opposed to vanishing rate in existing qudit Floquet code constructions.
\end{abstract}

\maketitle

Advantages offered by quantum information processing over classical processing are largely impeded by inevitable environmental noise, putting forth the need of an quantum error-correction procedure.
Recent developments in dynamical quantum error-correcting codes (QECC)~\cite{hastings2021dynamically,davydova2023floquet,fahimniya2023hyperbolic,gidney2021fault,haah2022boundaries,vuillot2021planar,paetznick2023performance,gidney2022benchmarking,hilaire2024enhanced,higgott2023constructions,aasen2023measurement,zhang2023x,dua2024engineering,alam2024dynamicallogicalqubitsbaconshor,bombin2023unifying,townsend2023floquetifying,ellison2023floquet,sullivan2023floquet,fu2024error,mcewen2023relaxing,kesselring2024anyon,davydova2023quantum,tanggara2024strategic} show that through an imposed dynamics on the system where information is encoded, one can obtain a fault-tolerant quantum memory, as well as perform fault-tolerant and error-correcting operations.
The most extensively studied family of dynamical QECCs called the ``Floquet codes''~\cite{hastings2021dynamically,davydova2023floquet,fahimniya2023hyperbolic,gidney2021fault,haah2022boundaries,vuillot2021planar,paetznick2023performance,gidney2022benchmarking,hilaire2024enhanced,higgott2023constructions,aasen2023measurement,zhang2023x,dua2024engineering,alam2024dynamicallogicalqubitsbaconshor,bombin2023unifying,townsend2023floquetifying,ellison2023floquet,sullivan2023floquet} offer advantages through performing a periodic sequence of two-body measurements that enables error-correction on a many-body system, relaxing hardware implementation requirements while also showing higher code threshold compared to static codes on certain platforms~\cite{gidney2021fault,hilaire2024enhanced}.
While there has been many results and proposals on static QECCs for higher dimensional quantum systems, or qudits~\cite{gottesman1997stabilizer,gottesman2016surviving,campbell2012magic,campbell2014enhanced,bullock2007qudit,aharonov1997fault,andrist2015error,nadkarni2021quantum,looi2008quantum,grassl2018quantum,miller2018propagation,schmidt2024error,ma2023non,watson2015qudit,anwar2014fast,duclos2013kitaev,hutter2015improved,watson2015fast,sabo2024trellis,brock2024quantumerrorcorrectionqudits}, our understanding of qudit dynamical QECCs is limited.
The use of qudits in quantum computation has been shown to exhibit advantages compared to qubit systems, such as: an improvement in compiling of multi-control gates exponential in the number of control registers~\cite{gokhale2019asymptotic,kiktenko2020scalable} (applicable to compilation of sub-routines for Grover's search~\cite{grover1996fastquantummechanicalalgorithm} and Shor's factoring~\cite{shor1999polynomial} algorithms), as well as a smaller number of qudits in a qubit-to-qudit circuit compilations~\cite{luo2014universal,luo2014geometry,wang2020qudits,kiktenko2023realizationquantumalgorithmsqudits,Nikolaeva_2024}, and higher yield and fidelity in magic-state distillation that tends to optimal asymptotically with qudit dimension~\cite{campbell2012magic,campbell2014enhanced}.
Moreover, the use of qudits has also been shown to increase security and efficiency in quantum cryptography~\cite{bruss2002optimal,cerf2002security,durt2003security,durt2004security,bradler2016finite,bouchard2017high}, increase noise-resistance and security of quantum communications~\cite{cozzolino2019high}, and enable simulation of high-dimensional quantum systems~\cite{sawaya2020resource,tacchino2021proposal,gonzalez2022hardware,meth2023simulating}
These advantages of qudits lead to qudit-based experimental realizations on various platforms, such as superconducting circuits~\cite{neeley2009emulation,blok2021quantum,morvan2021qutrit,roy2023two,goss2022high,fedorov2012implementation} photonics~\cite{lu2020quantum,wang2018multidimensional,chi2022programmable,reimer2019high}, trapped-ion~\cite{low2020practical,hrmo2023native,leupold2018sustained,ringbauer2022universal}, and circuit QED~\cite{brock2024quantumerrorcorrectionqudits}.
Thus to harness desirable features of both qudit quantum information and dynamical QECC fault-tolerantly, a further understanding on ways one can implement qudit processing tasks fault-tolerantly using dynamical QECCs. 

In this work we propose a family of dynamical qudit QECC called the ``qudit Floquet codes'', which protects logical quantum information by periodically performing two-qudit measurements on qudits placed either on a torus or a large family of hyperbolic surfaces.
Our qudit Floquet code construction is based on commutation conditions on the code's measurement sequence.
We show that a qudit Floquet code satisfying these conditions dynamically encodes $k$ logical qudits where $k$ grows linearly with the genus (i.e. the number of holes) of the lattice, although when viewed as a static subsystem code does not encode any logical information.
We also discuss a decoding strategy for our qudit Floquet code construction by using a ``space-time'' syndrome lattice to determine where and when errors occurred.
Lastly, we also provide a novel explicit construction of qudit Floquet codes for arbitrary prime dimension larger than two, based on the set of conditions that we proposed.

Compared to previous works on Floquet codes, our qudit Floquet code construction generalizes the qubit Floquet code constructions on a torus~\cite{hastings2021dynamically,davydova2023floquet,gidney2021fault} and on hyperbolic lattices~\cite{fahimniya2023hyperbolic,higgott2023constructions} to higher dimensional systems and provide more flexibility on the choice of measurements.
Moreover, our construction also generalizes the existing qudit Floquet code constructions on a torus~\cite{ellison2023floquet,sullivan2023floquet} to hyperbolic lattices, as well as allows more flexibility in the choice of measurements.
We also note that all of the aforementioned constructions of Floquet codes with explicit measurement settings (both qubit and qudit) are special cases of our construction.
This generalization to hyperbolic lattices results in a family of qudit Floquet codes with a rate of encoded logical qudits over the number of physical qudits that approaches $\frac{1}{2}$ asymptotically with the number of physical qudits in total and on the faces of the lattice, which is a significant improvement over the vanishing rate of existing qudit Floquet code constructions.

This paper is structured as follows:
In Section~\ref{sec:qudit_pauli_stabilizer} we review qudit Pauli group, qudit stabilizer group, and qudit stabilizer code.
Then in Section~\ref{sec:qubit_floquet_prelim} we review the ``honeycomb'' qubit Floquet code construction by Hastings and Haah~\cite{hastings2021dynamically}.
Section~\ref{sec:qudit_floquet_code} is the main part of the paper where we define our construction of a family of qudit Floquet codes.
We describe the sequence of measurements that defines the qudit Floquet codes and how the codespace dynamically changes as these measurements are performed in Section~\ref{sec:ISG}.
Particularly, we show how a qudit stabilizer group is updated after a Pauli measurement is performed in Theorem~\ref{thm:qudit_stabilizer_update}.
Then we use these update rules to show in Theorem~\ref{thm:general_qudit_floquet_code_conditions} that whenever a set of commutation conditions between the measurements are satisfied, the codespace evolves periodically such that at any given time, the number of encoded logical qudits scales linearly with the number of physical qudits.
For the rest of Section~\ref{sec:qudit_floquet_code}, we show that the encoding rate of the code approaches $\frac{1}{2}$ asymptotically with the number of physical qudits and the size of faces on the lattice and discuss how no logical information can be encoded when the qudit Floquet code is viewed as a subsystem code.
In Section~\ref{sec:logical_opers} we discuss in detail about how the logical operators of the qudit Floquet code can be constructed.
Then, in Section~\ref{sec:qudit_code_error_syndrome_decoding} we discuss a decoding strategy which involves a ``space-time syndrome lattice'' to identify where and when did error occurred.
Lastly in Section~\ref{sec:cirsq_qudit_floquet_code} we give a novel explicit instance of the qudit Floquet code for any prime dimension greater than two based on the conditions proposed in Theorem~\ref{thm:general_qudit_floquet_code_conditions}.

\section{Preliminaries}

\subsection{Qudit Pauli group and stabilizer group}\label{sec:qudit_pauli_stabilizer}

In this section, we review and set up notations for qudit Pauli group and qudit Stabilizer group in prime dimension $D$ (see~\cite{vourdas2004quantum,gheorghiu2014standard}).
The primitive components of the qudit Pauli group are the generalized Pauli $X$ and $Z$ operators defined by
\begin{equation}
\begin{gathered}
    X = \sum_{j=0}^{D-1} |j\oplus1\>\<j| \quad\textup{and}\quad
    Z = \sum_{j=0}^{D-1} \omega^j \ketbra{j}
\end{gathered}
\end{equation}
for $D$-th root of unity $\omega = e^{i\frac{2\pi}{D}}$, and thus $ZX = \omega XZ$ and $\{|0\>,\dots,|D-1\>\}$ is an orthonormal basis of $D$ dimensional complex vector space $\C^D$.
Note that $X$ and $Z$ are both unitary operators on $\C^D$, although they are not Hermitian except for $D=2$.
A (phaseless) \textit{qudit Pauli operator} is defined as $P_{a,b} = X^a Z^b$ where $a,b\in [D]:=\{0,\dots,D-1\}$.
These operators satisfy
\begin{equation}\label{eqn:qudit_pauli_commutation}
\begin{aligned}
    P_{a,b}P_{a',b'} = \omega^{c(P_{a,b},P_{a',b'})} P_{a',b'}P_{a,b} \;,
\end{aligned}
\end{equation}
where $c(\cdot,\cdot)$ is the \textit{commutation function} between two Paulis defined by $c(P_{a,b},P_{a',b'}) = -ab' + ba'$.
Here, the addition and products are all in mod $D$ (which henceforth is to be assumed for all integer arithmetics).
As $X$ and $Z$ does not commute, a power of $\omega$ factor also gives a distinct qudit Pauli operator $\omega^l P_{a,b} = \omega^l X^a Z^b = \omega^{l-ab} Z^bX^a$ for $l\in[D]$.
These operators makes up the \textit{qudit Pauli group}
\begin{equation}
\begin{gathered}
    \mathcal{P}_D = \{\omega^l P_{a,b} : a,b,l\in[D]\} \;.
\end{gathered}
\end{equation}
An $n$-qudit Pauli group is thus $\mathcal{P}_D^{\otimes n}$, for integer $n\geq1$.
For some $n$-qudit Pauli $P\in\mathcal{P}_D^{\otimes n}$, the projector onto its $\omega^a$-eigenspace is
\begin{equation}\label{eqn:qudit_pauli_eigenspace_projector}
    \Pi_a^{(P)} = \frac{1}{D} \sum_{j=0}^{D-1} (\omega^{-a} P)^j \;,
\end{equation}
as one can show that $(\Pi_o^{(P)})^2 = \Pi_o^{(P)}$ and that the sum of these projectors multiplied by their corresponding eigenvalues $\omega^0,\dots,\omega^{D-1}$ is equal to $P$:
\begin{equation}
\begin{aligned}
    \sum_{a=0}^{D-1} \omega^a \Pi_a^{(P)} &= \frac{1}{D} \sum_{j=0}^{D-1} P^j \Big(\sum_{a=0}^{D-1} \omega^{a(1-j)} \Big) \\
    &= \frac{1}{D} P \Big(\sum_{a=0}^{D-1} \omega^0 \Big) = P \;,
\end{aligned}
\end{equation}
since $\sum_{a=0}^{D-1} \omega^{a(1-j)} = \sum_{a=0}^{D-1} \omega^a = 0$ for all $j\neq 1$.

Note that for $D>2$, qudit Pauli $X$ and $Z$ are not Hermitian, hence all non-identity qudit Pauli $P_{a,b}$ do not correspond to a measurement observable.
However operators $\frac{P_{a,b}+P_{a,b}^\dag}{2}$ is Hermitian and has the same set of eigenvectors as $P_{a,b}$ with a one-to-one correspondence between their eigenvalues.
Namely, $P_{a,b}$ has eigenvectors $\{|\psi_j^{a,b}\>\}_j$ with eigenvalues $\{\omega^j\}_j$ if and only if $\frac{P_{a,b}+P_{a,b}^\dag}{2}$ has eigenvectors $\{|\psi_j^{a,b}\>\}_j$ with eigenvalues $\{\cos(j2\pi/D)\}_j$.
This can be shown by the relation
\begin{equation}
\begin{aligned}
    \frac{P_{a,b}+P_{a,b}^\dag}{2}|\psi_j^{a,b}\> &= \frac{\omega^j+\omega^{-j}}{2}|\psi_j^{a,b}\> \\
    &= \cos(j2\pi/D) |\psi_j^{a,b}\> \;,
\end{aligned}
\end{equation}
since $\frac{e^{i\theta}+e^{-i\theta}}{2} = \cos\theta$.
Thus throughout this paper, when we say that a qudit Pauli $P_{a,b}$ is measured with outcome $j$, the observable being measured effectively is $\frac{P_{a,b}+P_{a,b}^\dag}{2}$ where the outcome $j$ corresponds to its eigenvalue $\cos(j2\pi/D)$.

An \textit{$n$-qudit Stabilizer group} $\mathcal{S}$ is an Abelian subgroup of the $n$-qudit Pauli group $\mathcal{P}_D^{\otimes n}$.
A $n$-qudit Stabilizer group $\mathcal{S}$ generated by $m$ mutually commuting elements of $n$-qudit Paulis $\{g_1,\dots,g_m\}\subseteq \mathcal{P}_D^{\otimes n}$ is denoted as
\begin{equation}
\begin{aligned}
    \mathcal{S} = \<g_1,\dots,g_m\> = \Big\{ \prod_{j=1}^m g^{a_j} : a_j\in[D] \Big\} \;.
\end{aligned}
\end{equation}
If generators $g_1,\dots,g_m$ are independent elements of the $n$-qudit Pauli group $\mathcal{P}_D^{\otimes n}$, then there are $|\mathcal{S}|=D^m$ elements in the Stabilizer group $\mathcal{S}$ that it generates.
An $n$-qudit stabilizer group $\mathcal{S}$ generated by $g_1,\dots,g_m$ defines an \textit{$n$-qudit stabilizer code} $\mathcal{C}(\mathcal{S})$, which is the set of all $n$-qudit states $|\psi\>$ that is the common eigenstate of all elements of $\mathcal{S}$ with eigenvalue $1$.
For such a state $|\psi\>$ we say that it is \textit{stabilized} by $\mathcal{S}$ (and also stabilized by Pauli $s\in\mathcal{S}$).
Notationally it is expressed as
\begin{equation}\label{eqn:qudit_stabilizer_code}
    \mathcal{C}(\mathcal{S}) = \{|\psi\> : \forall s\in\mathcal{S}\,,\, s|\psi\>=|\psi\> \} \;.
\end{equation}
When $\mathcal{S}$ has a set of independent generators $g_1,\dots,g_m$, it has been shown in~\cite{gheorghiu2014standard} that $n$-qudit stabilizer code $\mathcal{C}(\mathcal{S})$ encodes $k=n-m$ logical qudits.
Similar to qubit stabilizer codes, qudit stabilizer codes has a code distance which can be obtained by taking the minimum weight (i.e. number of non-identity Paulis) over all logical operators of the code.
The logical operators are operators in the normalizer $\mathcal{N}(\mathcal{S})$ of stabilizer group $\mathcal{S}$ that themselves are not is $\mathcal{S}$, denoted by $\mathcal{N}(\mathcal{S})\backslash\mathcal{S}$.
So distance $d$ of code $\mathcal{C}(\mathcal{S})$ is
\begin{equation}
    d = \min\{\mathrm{weight}(L) : L\in \mathcal{N}(\mathcal{S})\backslash\mathcal{S}\} \;,
\end{equation}
where for an $n$-qudit Pauli $P\in\mathcal{P}_D^{\otimes n}$, we denote weight of $P = P_1\otimes\dots\otimes P_n$ by $\mathrm{weight}(P) = |\{i\in[n] : P_i\neq I\}|$ which is the number of non-identity qudit Paulis.
Finally, we can then obtain the code parameter $\llbracket n,k,d \rrbracket$ of $\mathcal{C}(\mathcal{S})$ which signifies the number of physical qudits, encoded logical qudits, and distance of the code.

\subsection{Floquet quantum error-correcting codes}\label{sec:qubit_floquet_prelim}

In this section we review Floquet quantum error-correcting codes, or simply Floquet codes.
Floquet codes is a family quantum error-correcting codes where the manner in which logical quantum information is encoded varies over time in a periodical manner.
The first Floquet code is proposed by Hastings and Haah in~\cite{hastings2021dynamically} which we refer to as Hastings-Haah Floquet code.
In the Hastings-Haah Floquet code, physical qubits in which logical information is encoded are placed on vertices of a hexagonal, or ``honeycomb'' lattice with periodic boundary conditions, tesselating a torus (see top Fig.~\ref{fig:surfaces}).
Each hexagonal face on the lattice is assigned either a green, red, or blue color such that no neighboring hexagons have the same color.
Each edge is also assigned the same color as the two hexagons that it is sharing a single qubit with.
For an illustration, see top Fig.~\ref{fig:qudit_honeycomb_lattice}.

As opposed to the conventional stabilizer topological quantum error-correcting codes where the stabilizer group $\mathcal{S}$ is defined by a fixed set of commuting Pauli operators on the qubits, a Floquet code is defined by a set of non-commuting weight-2 Pauli measurements on the edges of the lattice, called the \textit{check} measurements (or simply, ``checks'').
The check measurements are performed sequentially over multiple \textit{measurement rounds} (or simply, ``rounds'') with a period of three where in a given round $r$ a set of commuting checks $C_r$ are performed, although the checks $C_r$ of round $r$ does not commute with checks $C_{r-1}$ of the previous round and checks $C_{r+1}$ of the next round.
We can take the round $r=0\mod{3}$ checks to correspond to the green edges, round $r=0\mod{3}$ checks to correspond to the red edges, and round $r=0\mod{3}$ checks to correspond to the blue edges.
Since the checks are performed sequentially with period 3, therefore $C_r = C_{r+3}$.
This periodic sequence of checks is called a \textit{measurement schedule}.
We will define these concepts of Floquet codes more precisely in the next section when we define the qudit Floquet code.

Starting from an initial round $r=0$ and an initial stabilizer group $\mathcal{S}_0$ consisting only of the identity $I$, Hastings and Haah in~\cite{hastings2021dynamically} uses a set of rules for updating stabilizer group $\mathcal{S}_r$ to $\mathcal{S}_{r+1}$ after performing the checks $C_r$ on a codestate $|\psi\>$ of $\mathcal{S}_r$ (i.e. state $|\psi\>$ such that $s|\psi\>=|\psi\>$ for all $s\in\mathcal{S}_r$).
They show that for $r\geq3$ the stabilizer groups over the rounds evolves with period of three, namely $\mathcal{S}_r=\mathcal{S}_{r+3}$, where each stabilizer group $\mathcal{S}_r$ is the called instantaneous stabilizer group (ISG) of round $r$.
The ISG at any round $r\geq3$ contains weight-six stabilizers called the \textit{plaquette stabilizers} in its generator, each with support on a hexagon. 
In round $r=0\mod{3}$, the ISG generator additionally contains the weight-two Paulis corresponding to the green checks, whereas in round $r=1\mod{3}$ and $r=2\mod{3}$ the ISG contains the red and blue checks, respectively.
As such, the ISG evolves with period-3 from ISG $\mathcal{S}_\mathrm{green}$ with green check generator, to ISG $\mathcal{S}_\mathrm{red}$ with red check generator, to ISG $\mathcal{S}_\mathrm{blue}$ with blue check generator, then back to $\mathcal{S}_\mathrm{green}$ (see Fig.~\ref{fig:hastings_haah_floquet_ISG}).
Moreover they've shown that the stabilizer code $\mathcal{C}(\mathcal{S}_r)$ encodes two logical qubits for all $r\geq3$, although when seen as a qubit stabilizer subsystem code~\cite{poulin2005stabilizer} no logical information can be encoded\footnote{A qubit stabilizer subsystem code in this context is defined by a non-Abelian group $\mathcal{G}$ generated by all checks $C_0,C_1,C_2$ where one can obtain a stabilizer group $\mathcal{S}$ by taking the center of group $\mathcal{G}$. Particularly, Hastings and Haah have shown that the stabilizer code $\mathcal{C}(\mathcal{S})$ obtained in this manner encodes no logical information.}.
Many variants of Floquet code has since been proposed, including a ``CSS'' Floquet code~\cite{davydova2023floquet} where the checks consists only of qubit Pauli $X$ and $Z$, a generalization of Hastings-Haah Floquet code to hyperbolic lattices~\cite{fahimniya2023hyperbolic,higgott2023constructions}, as well as a ``planar'' Floquet code variant in which the lattice where qubits are placed has boundaries~\cite{vuillot2021planar,haah2022boundaries}.

\begin{figure}
    \centering
    \includegraphics[width=0.5\columnwidth]{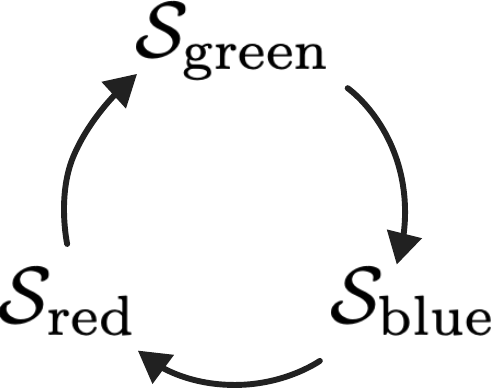}
    \caption{Periodic evolution of the instantaneous stabilizer group (ISG) in the Hastings-Haah Floquet code.}
    \label{fig:hastings_haah_floquet_ISG}
\end{figure}

\section{General Qudit Floquet Code}\label{sec:qudit_floquet_code}

In this section we define the construction of the family of qudit Floquet codes on arbitrary qudit of prime dimension $D$ and to three-colorable hyperbolic lattices.
Similar to the qubit honeycomb Floquet code, we start with a closed surface with genus $g\geq1$.
Surfaces with genus $g=1$ corresponds to a torus, whereas those with $g\geq2$ corresponds to hyperbolic surfaces with more than one ``holes'' (see Fig.~\ref{fig:surfaces}).
A lattice can be constructed on a surface by a $\{p,q\}$-\textit{tesselation} where $\{p,q\}$ is the Schl\"{a}fli symbol~\cite{blatov2010vertex} signifying a lattice consisting of faces with $p$ edges and vertices of degree $q$ (i.e. $q$ edges connected to each vertex).
We require these lattices to be three-colorable and for $p=3$, i.e. degree-3 vertices.
For a torus, we consider a $\{6,3\}$-tesselation which corresponds to hexagonal faces used in the honeycomb Floquet codes~\cite{hastings2021dynamically,gidney2021fault,davydova2023floquet} (see top Fig.~\ref{fig:qudit_honeycomb_lattice}).
For surfaces with genus $g\geq2$ we consider $\{p,3\}$-tesselations with even $p\geq8$.
For a tesselation on genus $g\geq2$ surfaces with $p=8$, we obtain lattices with octagonal faces that coincides with those considered in hyperbolic qubit Floquet code constructions in~\cite{fahimniya2023hyperbolic,higgott2023constructions}.
We note that for hyperbolic surfaces with genus $g\geq2$ a $\{p,q\}$-tesselation exists as long as $p$ and $q$ satisfy $(p-2)(q-2)>4$~\cite{albuquerque2009topological}.
Hence when fixing the degree $q=3$ we can obtain a lattice with $p$-gon faces for any $p>6$.

We then assign green, red, and blue colors to each faces of the lattice, as well as the edges such that the color of each edge is the same as the two faces it connects. 
The coloring scheme is illustrated in Fig.~\ref{fig:qudit_honeycomb_lattice} for a hexagonal and an octagonal lattice.
For the hexagonal lattice, this is the same three-coloring scheme as Floquet honeycomb codes in~\cite{gidney2021fault,davydova2023floquet}.
The coloring scheme for octagonal hyperbolic lattices are the same as those in qubit hyperbolic Floquet codes~\cite{fahimniya2023hyperbolic,higgott2023constructions}.

\begin{figure}
    \centering
    \includegraphics[width=0.9\columnwidth]{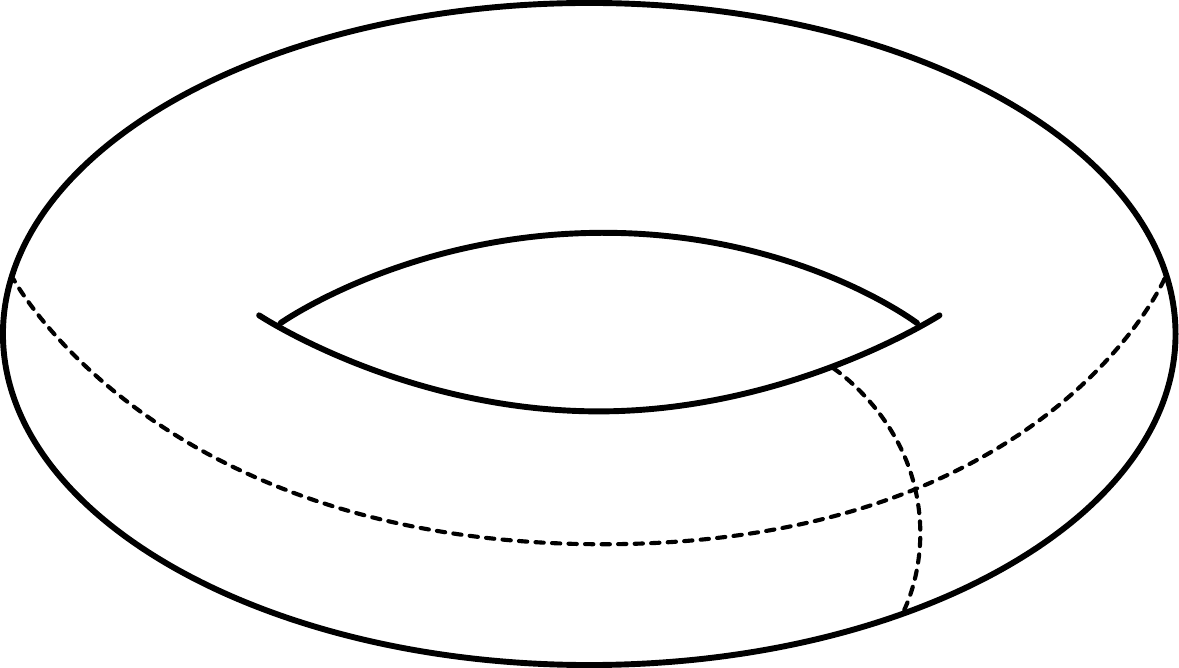}

    \vspace{0.8em}
    
    \includegraphics[width=1\columnwidth]{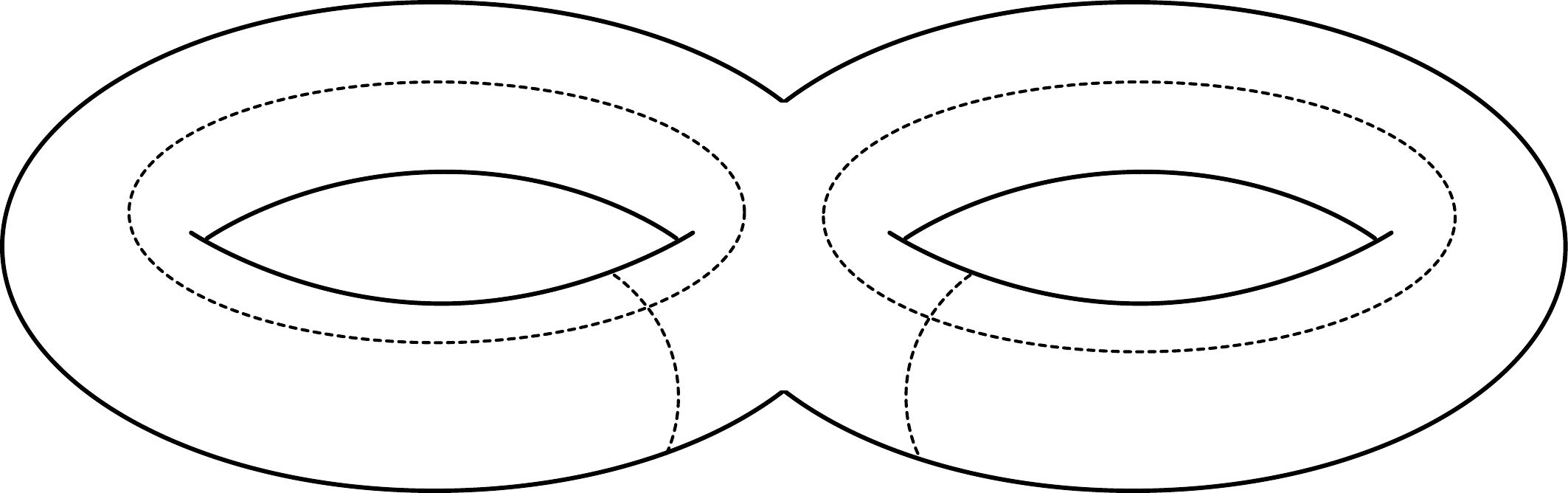}
    \caption{Top: A torus, i.e. surface with genus $g=1$. Bottom: Hyperbolic surface with genus $g=2$, i.e. two ``holes''.}
    \label{fig:surfaces}
\end{figure}

Similar to the qubit Floquet codes, how the qudit Floquet code encodes logical information is determined by how the check measurements are performed on the $n$ qudits over multiple rounds, which in turn dynamically changes the codespace.
Check measurements for the qudit Floquet code consists of weight-two qudit Pauli measurements corresponding to each edge.
Here we adopt a general approach in determining the type of check measurements by describing the lattice using a graph $G=(V,E)$ where its the set of vertices $V$ are the labels for the qudits on the lattice and the set of edges $E$ are the edges connecting the qudits.
We denote the edge connecting qudits $v$ and $u$ by $(v,u)\in E$.
The green, red, and blue three-coloring of the plaquettes and edges defines both the measurement schedule and the checks in each measurement schedule.

\begin{definition}[Check measurements on three-colorable lattices]\label{def:checks}
    Consider a three-colorable $\{p,3\}$-lattice on a genus-$g$ surface described by a graph $G=(V,E)$ where $p=6$ for $g=1$ and even $p\geq8$ for $g\geq2$.
    Assign green, red, and blue colors to the plaquettes and let $E_g,E_r,E_b \subset E$ be the set of green, red, and blue edges, respectively, where each color-$l$ edge connects two color-$l$ plaquettes.
    \textit{Green, red, and blue check measurements} are defined by
    \begin{equation}\label{eqn:general_checks}
    \begin{aligned}
        C_g = \{M_{(v,u)} : (v,u)\in E_g\} \\
        C_r = \{M_{(v,u)} : (v,u)\in E_r\} \\
        C_b = \{M_{(v,u)} : (v,u)\in E_b\}
    \end{aligned}
    \end{equation}
    respectively.
    For color $l\in\{g,r,b\}$, a check measurement in $C_l$ is of the form $M_{(v,u)} = [P_{v_l}]_v [P_{u_l}]_u \in C_l$ where $v_l,u_l\in[D]\times[D]$ (following the notation in Section~\ref{sec:qudit_pauli_stabilizer}) and for a qudit Pauli operator $P\in\mathcal{P}_D$, we denote $[P]_v$ as an $n$-qudit Pauli operator with identity everywhere except for a Pauli $P$ at the qudit labeled as $v$.
\end{definition}

Given these check measurements, we can now define a qudit Floquet code as follows.

\begin{definition}[Qudit Floquet code]\label{def:qudit_floquet_code}
    An $n$-\textit{qudit Floquet code} is defined on a three-colored lattice on a genus $g$ surface and a period-3 measurement schedule on checks as described in Definition~\ref{def:checks}:
    \begin{enumerate}
        \item In round $r=0\mod{3}$, measure all the green checks, 
        \item in round $r=1\mod{3}$ measure all the red checks, and 
        \item in round $r=2\mod{3}$ measure all the blue checks,
    \end{enumerate}
    such that after the first (finite number of) $T$ rounds, called the \textit{initialization rounds}, we have a code $\mathcal{C}_r$ of round $r\geq T$ that encodes $k$ logical qudits.
\end{definition}

\begin{figure}
    \centering
    \includegraphics[width=0.8\columnwidth]{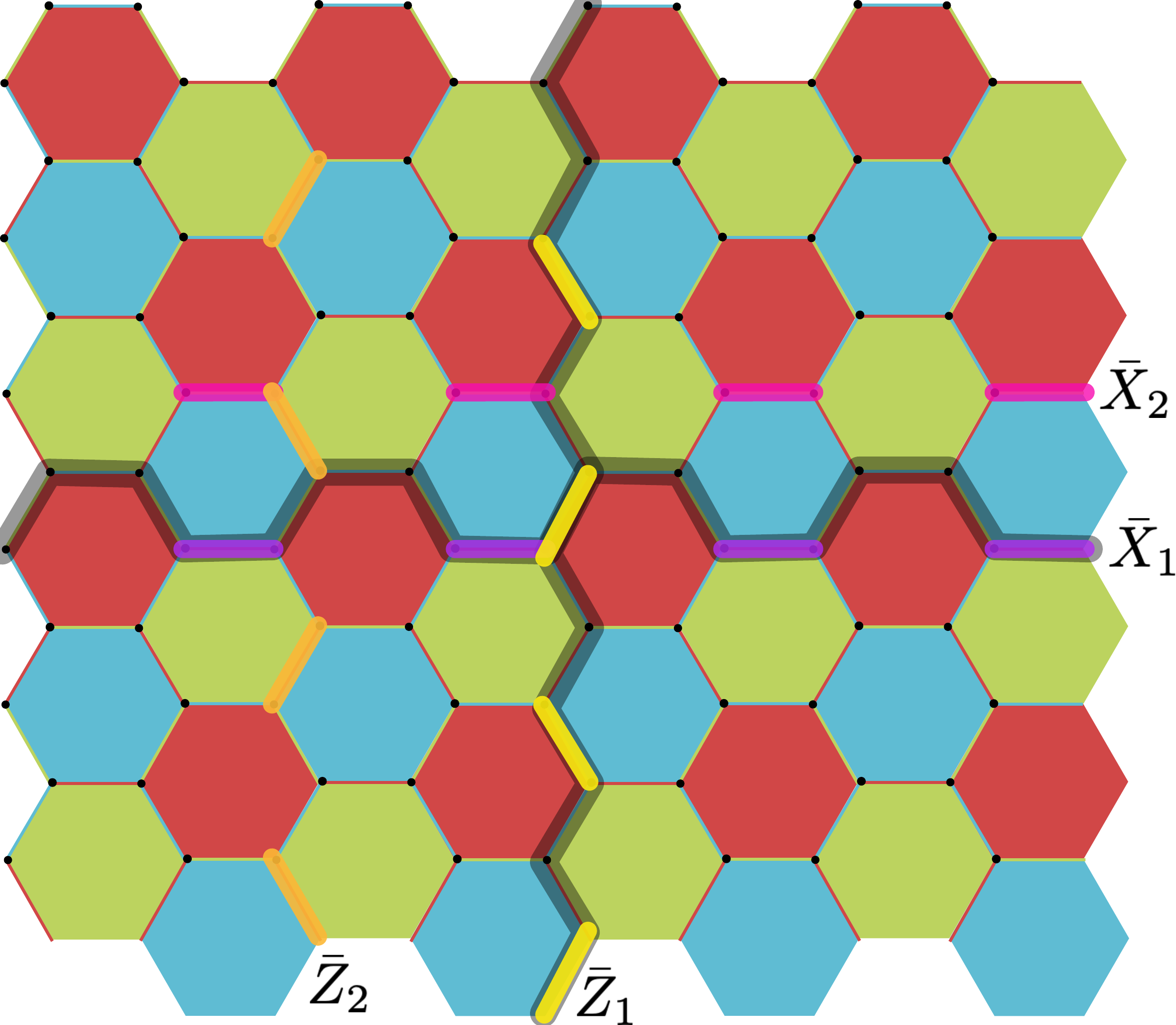}

    \vspace{0.8em}
    
    \includegraphics[width=0.8\columnwidth]{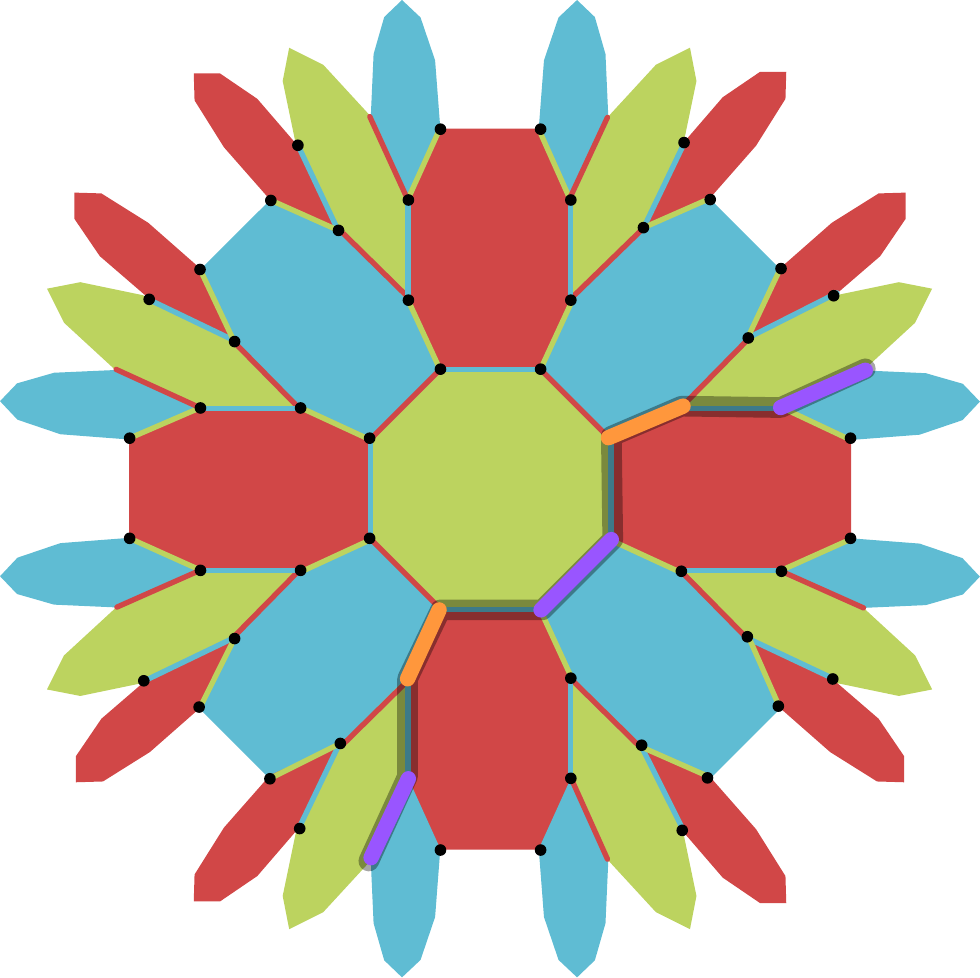}
    
    \caption{Top figure: Illustration of a $96$-qudit Floquet code on a hexagonal ``honeycomb'' lattice with degree 3 on the surface of a torus.
    The bottom and top boundaries represent the same set of edges, similarly for its left and right boundaries.
    Stabilizer group evolves with period 3, each encoding two logical qudits, is obtained by periodically measuring checks on green, red, and blue edges satisfying condition stated in Theorem~\ref{thm:general_qudit_floquet_code_conditions}.
    Edges highlighted in orange, yellow, purple and pink on non-trivial loops (e.g. chain of grey edges in vertical and horizontal directions) illustrate a representative of $\Bar{Z}_1,\Bar{Z}_2,\Bar{X}_1,\Bar{X}_2$ logical operators, respectively, for stabilizer group in round $r=0\mod{3}$.
    As can be seen in the figure (and can be verified for other rounds), the minimum weight of the logical operators in each round for the $96$-qudits Floquet code is $8$.
    Hence, the honeycomb qudit Floquet code has a distance $d=8$ and a code parameter $\llbracket 96,2,8\rrbracket$.
    Bottom figure: Illustration of qudit Floquet code on a part of a hyperbolic lattice with degree three and genus $g$.
    The stabilizer groups of the code, evolving with period three, encodes $2g$ logical qudits whenever the checks satisfy conditions in Theorem~\ref{thm:general_qudit_floquet_code_conditions}.
    A pair of commuting logical operators $\Bar{L},\Bar{L}'$ in round $r=0$ on a non-trivial loop (chain of grey edges) with support on orange edges and purple edges, respectively.
    }
    \label{fig:qudit_honeycomb_lattice}
\end{figure}

\subsection{Instantaneous stabilizer groups}\label{sec:ISG}

Now we describe how the codespaces are obtained from the measurement schedule.
Between checks of round $r$ and $r+1$ we obtain a qudit stabilizer group $\mathcal{S}_r$ of round $r$ called an \textit{instantaneous stabilizer group (ISG)}, since we are concerned only with Pauli check measurements.
ISG $\mathcal{S}_r$ defines a qudit stabilizer code $\mathcal{C}(\mathcal{S}_r)$ (see eqn.~\eqref{eqn:qudit_stabilizer_code}).
We start with round $r=-1$ with a stabilizer group $\mathcal{S}_{-1}$ which contains only the identity operator $I$, then the check measurements in the subsequent rounds will populate the ISGs $\mathcal{S}_0$, $\mathcal{S}_1$, and so on.
Namely, the check measurements in round $r$ updates ISG $\mathcal{S}_{r-1}$ to $\mathcal{S}_r$.

To do this we need a general set of rules to update a qudit stabilizer group $\mathcal{S}$ after a measurement of qudit Pauli observable $P$ with an outcome $o\in[D]$ to a new qudit stabilizer group $\mathcal{S}'$.
Here we generalize the qubit stabilizer update rules first shown in~\cite{gottesman1998heisenberg}\footnote{The qubit stabilizer update rules are used in~\cite{hastings2021dynamically,townsend2023floquetifying,delfosse2023spacetime,fu2024error} for Floquet codes and other dynamical error-correction schemes.} to qudit stabilizers. 

\begin{theorem}\label{thm:qudit_stabilizer_update}
    Update rules for an initial stabilizer $\mathcal{S}=\<g_1,\dots,g_m\>$ to stabilizer $\mathcal{S}'$ after a measurement of Pauli observable $P$ with outcome $o\in[D]$ :
    \begin{enumerate}
        \item If $\omega^a P\in\mathcal{S}$ for some $a\in[D]$, then $\mathcal{S}' = \mathcal{S}$, i.e. the stabilizer stays the same.
        In this case, outcome $o$ is deterministic.
        
        \item If $\omega^a P$ is not in $\mathcal{S}$ for all $a\in[D]$ but $P$ commutes with all $s\in\mathcal{S}$, then $\mathcal{S}' = \<\omega^{-o}P,g_1,\dots,g_m\>$.
        i.e. we simply add $\omega^{-o}P$ to the set of generators.
        
        \item If $P$ does not commute with each element in a subset $W = \{g_1,\dots,g_l\}$ (for $l\leq m$) of the generator set $G=\{g_1,\dots,g_m\}$ for $\mathcal{S}$, then $\mathcal{S}' = \<W',\omega^{-o}P,G\backslash W\>$.
        Here, $W'$ is a set of $l-1$ stabilizers from $\mathcal{S}$ such that $\mathcal{S}=\<W',g^*,G\backslash W\>$ for some arbitrarily chosen $g^*\in W$.
        Namely, here we replace generator $g^*$ with $\omega^{-o}P$, replace the rest of $W$ with $W'$, and keep the generators that commute with $P$.
    \end{enumerate}
\end{theorem}

The proof of Theorem~\ref{thm:qudit_stabilizer_update} is given in Appendix~\ref{app:qudit_stabilizer_update_proof}.
Using these rules we can then obtain the ISGs $\mathcal{S}_r$ in any round $r$ by applying these rules for the check measurements of round $r$ on ISG $\mathcal{S}_{r-1}$ as the initial qudit stabilizer group.
As we will show below, we obtain ISGs with a period of three after the initialization rounds when the checks in eqn.~\eqref{eqn:general_checks} satisfy certain commutation conditions.
Namely we have three distinct ISGs $\mathcal{S}_r$, $\mathcal{S}_{r+1}$, and $\mathcal{S}_{r+2}$ for round $r \mod{3}$ after initialization.
Moreover, each ISG encodes two logical qudits and all three ISGs share a common subset of generators which is used as the stabilizers with eigenvalues that correspond to error syndromes inferred from the check measurements.
These common ISG generators are the weight-six stabilizers on each of the hexagonal faces on the lattice, known as the \textit{plaquette stabilizers} (or simply as ``plaquette'', following the terminology used in the Floquet code literature~\cite{hastings2021dynamically,davydova2023floquet,fahimniya2023hyperbolic} and topological codes in general, see e.g. lecture notes~\cite{Browne_2014}).
Now we state the commutation conditions for checks in eqn.~\eqref{eqn:general_checks} to obtain a qudit Floquet code with ISGs that encodes $2g$ logical qudits.

\begin{theorem}\label{thm:general_qudit_floquet_code_conditions}
    An $n$-qudit Floquet code on a three-colorable $\{p,3\}$-lattice on a genus-$g$ surface (for $p=6$ for $g=1$ and even $p\geq8$ for $g\geq2$) encodes $k=2g$ logical qudits if the following conditions for checks $C_g,C_r,C_b$ are satisfied:
    \begin{enumerate}
        \item It holds that
        \begin{equation}
        \begin{aligned}
            c(P_{v_g},P_{v_b}) = - c(P_{u_g},P_{u_b}) \;,\;\forall (v,u)\in E_g \\
            c(P_{v_r},P_{v_g}) = - c(P_{u_r},P_{u_g}) \;,\;\forall (v,u)\in E_r \\
            c(P_{v_b},P_{v_r}) = - c(P_{u_b},P_{u_r}) \;,\;\forall (v,u)\in E_b
        \end{aligned} \;.
        \end{equation}

        \item For all vertex $v\in V$,
        \begin{equation}\label{eqn:checks_noncommute}
        \begin{aligned}
            c(P_{v_g},P_{v_r}) \neq 0 \\
            c(P_{v_r},P_{v_b}) \neq 0 \\
            c(P_{v_b},P_{v_g}) \neq 0 
        \end{aligned} \;.
        \end{equation}
        
        \item For all edge $(v,u)\in E$,
        \begin{equation}
        \begin{aligned}
            [P_{v_g} P_{v_r} P_{v_b}]_v \, [P_{u_g} P_{u_r} P_{u_b}]_u = I^{\otimes n} \;.
        \end{aligned}
        \end{equation}
    \end{enumerate}
\end{theorem}

Detailed proof of Theorem~\ref{thm:general_qudit_floquet_code_conditions} above is given in Appendix~\ref{app:general_qudit_floquet_code_conditions}.
Here we discuss the intuition behind the conditions and give a brief outline of the proof.
In the proof outline given below, we mainly focus on how the measurement schedule of the checks satisfying conditions in Theorem~\ref{thm:general_qudit_floquet_code_conditions} for the qudit Floquet code for simplicity.
A generalization to hyperbolic lattices resulting from $\{p,3\}$ tesselation for even $p\geq8$ is straightforward and can be found in Appendix~\ref{app:general_qudit_floquet_code_conditions}.

Condition 2 of Theorem~\ref{thm:general_qudit_floquet_code_conditions} simply states that any two checks that shares a qudit in their supports do not commute.
Condition 1 on the other hand implies that a color-$l$ check on $(v,u)$ commutes with the product of two checks of color $l'$ with support on qudit $v$ and qudit $u$.
Condition 3 states that given an edge $(v,u)$ the product of checks with support on $v$ or $u$ or both must have an identity $I$ on each qudit $v$ and $u$.
Hence the product of all checks must be $I^{\otimes n}$ as a consequence.

\begin{figure*}
    \centering
    \includegraphics[width=0.6\textwidth]{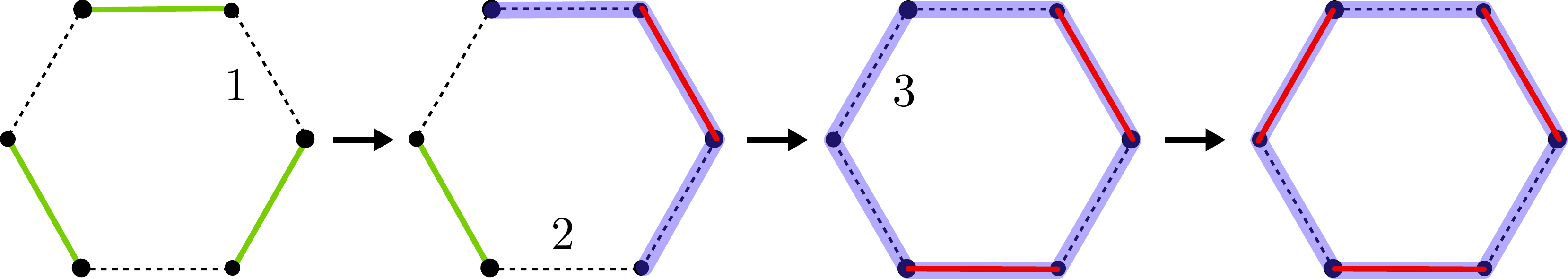}
    
    \vspace{0.5em}
    
    \includegraphics[width=1\textwidth]{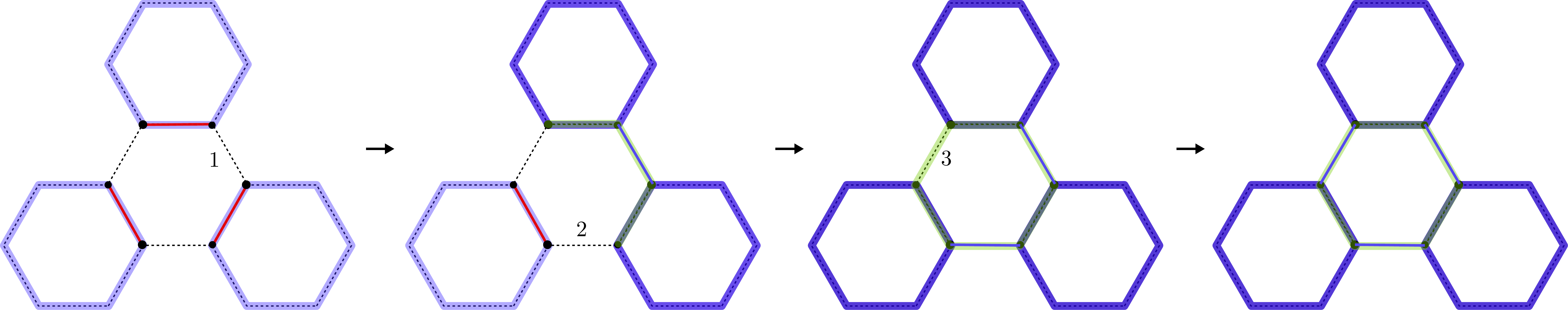}
    \caption{First row: Update sequence on a blue plaquette in round $r=1$ where we measure three red checks round it.
    Second row: Update sequence on a green and three blue plaquettes in round $r=2$ where we measure the blue checks surrounding the green plaquette.
    Each step in the sequence illustrates the set of stabilizer generators as colored connected edges and the check measured at that step as the numbered edge, which leads to the updated generators illustrated in the next step of the sequence.
    Updates are performed according to the update rules in Theorem~\ref{thm:qudit_stabilizer_update}.
    Instantaneous stabilizer groups (ISGs) in each round of the qudit Floquet code is obtained by performing one of the types of update sequence as the three checks are surrounding each plaquette are measured (the first or second row in the figure).
    For example, the second step in the top row shows three generators: a weight-2 generator (green edge), another weight-2 generator (red edge), and a weight-4 generator (blue edges).
    In this step a check on the edge labeled by ``2'' is measured, leading to an updated generators in the next step of the diagram: a weight-6 generator corresponding to a blue plaquette stabilizer (blue edges) and two weight-2 generators corresponding to two red checks (two red edges).
    In the second row, the green plaquette is updated as in the first row.
    However three blue plaquettes surrounding the middle green plaquette are also updated as blue checks are performed, where initial blue plaquettes are in light-blue color and updated blue plaquettes are in dark-blue color.}
    \label{fig:plaquette_update}
\end{figure*}

Let $C_l$ and $A_l$ be the set of all checks and all plaquette stabilizers of color $l\in\{g,r,b\}$, respectively.
When the checks satisfy these three conditions in Theorem~\ref{thm:general_qudit_floquet_code_conditions}, ISGs $\mathcal{S}_r$ after check measurement in round $r\geq0$ evolve as:
\begin{enumerate}
    \item At round $r=-1$ we start with ISG $\mathcal{S}_{-1}$ consisting only of the identity $I^{\otimes n}$.

    \item At round $r=0$ the ISG is $\mathcal{S}_0 = \<C_g\>$.

    \item At round $r=1$ the ISG is $\mathcal{S}_1 = \<A_b',C_r\>$.

    \item At round $r=2$ the ISG is $\mathcal{S}_2 = \<A_g',A_b,C_b\>$.

    \item At round $r=3$ the ISG is $\mathcal{S}_3 = \<A_r',A_g,A_b,C_g\>$.

    \item At round $r=4$ the ISG is $\mathcal{S}_4 = \<A_r,A_g,A_b,C_r\>$.
    
    \item At round $r>4$ the ISG is $\mathcal{S}_r = \<A_r,A_g,A_b,C_l\>$ where $l=g$ if $r=0\mod{3}$ and $l=r$ if $r=1\mod{3}$ and $l=b$ if $r=2\mod{3}$.
\end{enumerate}
Here $A_l'$ is the set of ``unformed'' plaquette stabilizers which consists of plaquette stabilizers of color $l$ that will be changed in a later initialization round to $A_l$.
The reason for this distinction is that plaquette stabilizers in $A_l'$ are non-commuting with some checks, hence needs to be updated to plaquette stabilizers $A_l$ which commute with all checks and all other plaquette stabilizers.
Here we assume throughout the initialization round that all check measurements has an outcome of $0$ corresponding to the eigenvalue $\omega^0=+1$ of the check\footnote{In general, if the outcome of a check in round $r$ is $o\neq 0$, then some of the checks and plaquette stabilizers that are updated from that check measurement (according to update rules in Theorem~\ref{thm:qudit_stabilizer_update}) in ISG $\mathcal{S}_r$ will be multiplied by a phase factor of $\omega^{-o}$.}.
Throughout the rest of this paper we will use outcome $0$ and $+1$ interchangeably, as will be made clear from the context.

Note that the non-commutation conditions between checks (condition 2 of Theorem~\ref{thm:general_qudit_floquet_code_conditions}) induces the periodic dynamics of the ISG as the checks in round $r$ are non-commuting with the checks in round $r-1$.
In each round $r$ we apply the appropriate update rule in Theorem~\ref{thm:qudit_stabilizer_update} to update the generators of ISG $\mathcal{S}_{r-1}$ and obtain the ISG $\mathcal{S}_r$ of round $r$ which induces a period-3 evolution of the ISGs, i.e. $\mathcal{S}_r=\mathcal{S}_{r+3}$ for $r\geq4$ after the initialization rounds (the first five rounds of check measurements).

We start with stabilizer group $\mathcal{S}_{-1}$ consisting only of the $n$-qudit identity Pauli $I^{\otimes n}$.
After measuring the green checks $C_g$ in the first round $r=0$, we simply apply update rule 2 and add all green checks to the stabilizer generator to obtain round-0 ISG $\mathcal{S}_0=\<C_g\>$. 
In the rest of initialization rounds $r\in\{1,2,3,4\}$, we do one of the two types of update sequences on the three checks surrounding a plaquette illustrated in Fig.~\ref{fig:plaquette_update} to form the plaquette stabilizers.

The first type of update sequence is illustrated in the first row of Fig.~\ref{fig:plaquette_update} for round $r=1$.
Here we obtain a blue plaquette stabilizer as we sequentially measure red checks on edges labeled by 1, 2, 3 which are going to be in the generator of ISG $\mathcal{S}_1$.
As ISG $\mathcal{S}_0$ is generated only by the green checks $C_g$, a blue plaquette is surrounded by three green checks at the start of this round.
Measuring red check 1 and 2 invokes update rule 3 as the green checks and red checks are non-commuting.
Measuring red check 1 removes two green checks and adds a weight-4 Pauli which is a product of the two green checks.
Then, measuring red check 2 removes the weight-4 Pauli and the remaining green check and adds a weight-6 Pauli which is a product of the weight-4 Pauli and the remaining green check.
These weight-4 Pauli and weight-6 Pauli are the blue colored edges in second and third step in the top diagram of Fig.~\ref{fig:plaquette_update}, respectively.
These two updates are due to the second line of condition 2 in Theorem~\ref{thm:general_qudit_floquet_code_conditions} which says that a red check commutes with the product of two green checks which it has common support with.
By the same reasoning, red check 3 commutes with the weight-6 Pauli in blue as it is simply the product of the three green checks.
Thus we simply apply update rule 2 and add this red check to the generators.
By repeating the same update for the three red checks surrounding each blue plaquette, we obtain a weight-6 Pauli for each blue plaquette which is precisely the ``unformed'' plaquette stabilizers $A_b'$ in ISG $\mathcal{S}_1$.
Thus we obtain $\mathcal{S}_1=\<A_b',C_r\>$.
A generalization to hyperbolic lattices from a $\{p,3\}$ tesselation on surfaces of genus $g\geq2$ is straightforward as we require $p\geq8$ to be even.
In this case we label the red edges as $1,\dots,p/2$ then perform update rule 3 for the first $p/2-1$ red checks and update rule 2 for the last red check to obtain unformed blue plaquette stabilizers $A_b'$.

The second type of update sequence, illustrated in the second row of Fig.~\ref{fig:plaquette_update}, is performed in round 2, 3, and 4.
This update sequence adds unformed plaquette stabilizers with color of the next round's check, as well as updates the unformed plaquette stabilizers $A_l'$ obtained in the previous round to the initialized plaquette stabilizers $A_l$.
For example, in round $r=2$ we start with three unformed blue plaquette stabilizers and three red checks surrounding a green plaquette (first sequence in the second row of Fig.~\ref{fig:plaquette_update}).
Blue check 1 does not commute with two unformed blue plaquette stabilizers as well as two red checks.
Thus they will invoke update rule 3 which gives us an updated blue plaquettes and a weight-4 Pauli.
The weight-4 Pauli is simply the product of two red checks, but an updated blue plaquette stabilizer is a product between an unformed blue plaquette and a red check on its boundary.
This updated blue plaquette commute with blue check 1 since the product of Paulis of all three checks with support on a given qudit on vertex $v$ is proportional to the identity $I$ by condition 3 of Theorem~\ref{thm:general_qudit_floquet_code_conditions}.
These updated blue plaquettes are illustrated as dark blue edges in second row of Fig.~\ref{fig:plaquette_update}, where in the second step the top blue plaquette is the product between the top unformed blue plaquette (top light blue in step 1) and the top red check.
The remaining unformed blue plaquettes are updated in the same way, whereas for blue check 3 it can be shown using the last line of condition 3 that it does commute with the two updated blue plaquette stabilizers.

By repeating this for each green plaquette and three blue plaquettes surrounding it, we obtain the set of unformed green plaquette stabilizers $A_g'$ and initialized blue plaquette stabilizers $A_b$ where each plaquette in $A_b$ is the product between an blue plaquettes with three red checks surrounding it.
Thus we obtain ISG $\mathcal{S}_2=\<A_g',A_b,C_b\>$.
For qudit Floquet code on hyperbolic lattices with $\{p,3\}$ tesselation we again label the blue edges surrounding a given green plaquette as $1,\dots,p/2$.
The surrounding unformed blue plaquettes are updated as we obtain the unformed green plaquette in the first $p/2-1$ blue checks, whereas we can simply add the last blue check to the stabilizer group generator as they commute with the updated blue plaquettes and the unformed green plaquette.

Updates in initialization round 3 is performed in the same way as round 2 where we update unformed plaquette stabilizers $A_g'\mapsto A_g$ and obtain the unformed red plaquettes $A_r'$, while keeping the blue plaquettes $A_b$.
Whereas updates in initialization round 4 simply updates the unformed red plaquettes $A_r'\mapsto A_r$ while keeping the green and blue plaquettes $A_g,A_b$.
The reason why plaquette stabilizers $A_g$ and $A_b$ remains as a stabilizer generator is due to the fact that all checks commute with updated blue plaquette stabilizers, which again can be shown using condition 1.
Thus in round $r>4$ then the plaquette stabilizers remain in the ISGs and check measurements simply replaces the checks from the previous round in the stabilizer group generator.

\subsection{Encoding rate}

The qudit Floquet code construction specified in Theorem~\ref{thm:general_qudit_floquet_code_conditions} have a nice feature where its encoding rate $\frac{k}{n}$ approaches $\frac{1}{2}$ as the number of physical qudits $n$ and the plaquette stabilizer weight $p$ grows larger, as stated in the following proposition.

\begin{proposition}
    Encoding rate of an $n$-qudit Floquet code on a three-colorable $\{p,3\}$-lattice on a genus-$g$ surface for $p\geq6$ and $g\geq2$ with checks $C_g,C_r,C_b$ satisfying conditions in Theorem~\ref{thm:general_qudit_floquet_code_conditions} has an encoding rate of 
    \begin{equation}
        \frac{k}{n} = \frac{1}{2}-\frac{3}{p}+\frac{2}{n} \;.
    \end{equation}
    Hence the encoding rate is $\frac{k}{n} \rightarrow \frac{1}{2}$ as $n\rightarrow\infty$ and $p\rightarrow\infty$.
\end{proposition}

\begin{proof}
    Consider the relation $n_p=n\frac{q}{p}$ (see~\cite{albuquerque2009topological}) between the number of plaquettes $n_p$ and the number of qudits $n$ and $\{p,q\}$-tesselation parameters $p$ and $q$ where $p\geq8$ is even and $q=3$.
    Then by expressing the number of logical qudits $k$ in terms of $n_p$ and $p$ as 
    \begin{equation}
        k = \frac{n_pp}{6} - n_p + 2
    \end{equation}
    (see eqn.~\eqref{eqn:relation_logical_qudits_plaquettes_hyeprbolic} in Appendix~\ref{app:general_qudit_floquet_code_conditions}) and substituting in $n_p=n\frac{3}{p}$, we obtain
    \begin{equation}
    \begin{aligned}
        \frac{k}{n} &= \frac{n_pp/6 - n_p + 2}{n} \\
        &= \frac{n/2 - 3n/p + 2}{n} \\
        &= \frac{1}{2}-\frac{3}{p}+\frac{2}{n} \;,
    \end{aligned}
    \end{equation}
    as claimed.
\end{proof}

We note that a similar rate was shown for qudit hyperbolic Floquet code in~\cite{higgott2023constructions} in terms of a different parameterization for the lattice construction.
Although similar asymptotic encoding rate can also be achieved by hyperbolic stabilizer codes, such as those constructed in~\cite{albuquerque2009topological,delfosse2013tradeoffs}, one gets a more demanding hardware requirements as the weight of stablizier grows as these stabilizers are directly measured.
On the other hand in the qudit Floquet code, one can simply infer the eigenvalues of these weight-$p$ stabilizers by measuring the weight-2 checks sequentially.

\subsection{Subsystem code and gauge group of qudit Floquet code}

Similar to the qubit honeycomb Floquet code~\cite{hastings2021dynamically}, a qudit Floquet code can be thought of as a subsystem code~\cite{poulin2005stabilizer} on a $\{p,q\}$ three-colorable lattice with gauge group $\mathcal{G}$ generated by the checks, i.e. $\mathcal{G} = \<C_g,C_r,C_b\>$.
We now identify the center of gauge group $\mathcal{G}$, denoted by $\mathcal{Z}(\mathcal{G})$ which contains all operators in $\mathcal{G}$ that commute with all operator $g\in\mathcal{G}$.
The center of gauge group $\mathcal{G}$ corresponds to the stabilizer group $\mathcal{S}=\mathcal{Z}(\mathcal{G})$ that defines the code space of the subsystem code.
Here we will show that as in the case of qubit honeycomb Floquet code~\cite{hastings2021dynamically}, that the qudit Floquet code with checks satisfying conditions in Theorem~\ref{thm:general_qudit_floquet_code_conditions} does not encode any logical information when treated as a subsystem code.
However, we note that qudit topological subsystem codes with checks that do not satisfy the conditions in Theorem~\ref{thm:general_qudit_floquet_code_conditions} could still encode some logical information.
One of such qudit topological subsystem code construction is given in~\cite{ellison2023pauli}, which stabilizer group contains some non-local stabilizers.

Stabilizer group $\mathcal{S}=\mathcal{Z}(\mathcal{G})$ consists of product of checks on a cycle in the lattice, e.g. product of six checks around a plaquette and product of all checks along a non-contractible loop around the torus.
Similar to the proof of Theorem~\ref{thm:general_qudit_floquet_code_conditions} above (see Appendix~\ref{app:general_qudit_floquet_code_conditions}), the operators obtained by taking the products on a cycle commute with all checks, hence they must be in $\mathcal{S}$.
Moreover since checks of different color do not commute with each other, product of checks that do not form a cycle does not commute with some element of $\mathcal{G}$.
Thus $\mathcal{S}$ consists of product of checks along each cycle.

Now note that by condition 3 of Theorem~\ref{thm:general_qudit_floquet_code_conditions} that the product of all checks is equal to identity, and therefore there are $|E|-1 = \frac{3n}{2}-1$ independent generators of $\mathcal{G}$ (since there are $3n/2$ edges, see Appendix~\ref{app:general_qudit_floquet_code_conditions}).
As shown in the proof of Theorem~\ref{thm:general_qudit_floquet_code_conditions} that the product of all plaquette operators (each obtained by taking the product of the six checks around them) is equal to identity, therefore there can only be $n_p-1$ plaquette operators in the independent generator set for $\mathcal{S}$.
In the case of qudit Floquet code on a torus (genus $g=1$ surface) we have $n_p=\frac{n}{2}$, whereas for a hyperbolic qudit Floquet code on a surface with genus $g\geq2$ with $\{p,3\}$-tesselation we have $n_p=\frac{3n}{p}$.
Additionally we also have $2g$ independent non-contractible loops as generators of $\mathcal{S}$, which is also in the independent generator set for $\mathcal{S}$ (as these cannot be generated by the plaquette operators).

Thus we have $\frac{3n}{p}-1$ plaquette operators and $2g$ non-contractible loops in the independent generator set for $\mathcal{S}$.
By using the relation $n_p = n\frac{3}{p}$ and eqn.~\eqref{eqn:relation_logical_qudits_plaquettes_hyeprbolic}, we have 
\begin{equation}
\begin{aligned}
    \frac{3n}{p}-1 + 2g &= \frac{3n}{p}-1 + \frac{n_pp}{6}-n_p+2 \\
    &= \frac{3n}{p}-1 + \frac{n}{2}-\frac{3n}{p}+2 \\
    &= \frac{n}{2}+1
\end{aligned}
\end{equation}
independent generators for $\mathcal{S}$.
Now since there are $\frac{3n}{2}-1$ independent generators of the gauge group $\mathcal{G}$ and $\frac{n}{2}+1$ independent generators of stabilizer group $\mathcal{S}$ we have $(3n/2-1) - (n/2+1) = n-2$ logical operators for the gauge qudits.
Since each gauge qudit has a logical $X$ and logical $Z$ operators, we have a total of $n/2-1$ gauge qudits.
Now as the total number of independent stabilizer generators plus the number of gauge qudits is equal to the number of physical qudits $n = (n/2-1) + (n/2+1)$, we can conclude that the subsystem code does not encode any logical information.

\subsection{Subfamilies of the qudit Floquet code}

As the reader may note, so far we have not explicitly express the Pauli operators corresponding to check measurements and plaquette stabilizers of a qudit Floquet code obtained from conditions in Theorem~\ref{thm:general_qudit_floquet_code_conditions}.
In Section~\ref{sec:cirsq_qudit_floquet_code} we construct a novel subfamily of qudit Floquet codes called the $(\vcirc,\vsq)$ qudit Floquet code satisfying these conditions by imposing an additional feature to a lattice obtained from any $\{p,3\}$ tesselation for even $p\geq8$.
In this construction

We also note that the honeycomb qudit Floquet code construction proposed in~\cite{ellison2023floquet} is an instance of our construction as specified by conditions in Theorem~\ref{thm:general_qudit_floquet_code_conditions}.
In~\cite{ellison2023floquet}, the check measurement corresponding to a given edge is determined by their direction on the lattice.
By considering the cardinal directions of the lattice illustrated in Fig.~\ref{fig:qudit_honeycomb_lattice}, the edges with direction from south-west to north-east are labeled by $x$, and edges with direction from south-east to north-west are labeled by $x$, and horizontal edges are labeled by $z$.
Consider an edge $(v,u)$.
If it is labeled by $x$ correspond to a $[X]_v [X]_u$ check, whereas if it has a label $y$ then it corresponds to a $[(XZ^D)^\dag]_v [(XZ^D)^\dag]_u$ check, and if it has label $z$ then it corresponds to a $[Z^D]_v[Z^D]_u$ check for qudit Floquet code with dimension $D\geq2$.
Note that condition 1 of Theorem~\ref{thm:general_qudit_floquet_code_conditions} is satisfied as for any edge $(v,u)$ of color $l$, the checks of color $l'$ with support on $v$ and on $u$ satisfy the corresponding commutation relation in condition 1, e.g. $c([X]_v,[(XZ^D)^\dag]_v) = -c([X]_v,[XZ^D]_v) = -c([X]_u,[Z^D]_u)$ for an $x$ edge.
Condition 2 is also satisfied since any pair of Paulis from $X,(XZ^D)^\dag,Z^D$ does not commute.
Lastly, condition 3 is satisfied since $X(XZ^D)^\dag Z^D \otimes X(XZ^D)^\dag Z^D = I\otimes I$.

\section{Logical operators of the qudit Floquet code}\label{sec:logical_opers}

For a qudit Floquet code on a lattice with genus $g$ and checks that satisfy the conditions in Theorem~\ref{thm:general_qudit_floquet_code_conditions}, there are $2g$ logical operators $\Bar{X}_1,\Bar{Z}_1,\Bar{X}_2,\Bar{Z}_2 , \dots, \Bar{X}_{2g},\Bar{Z}_{2g} \in \mathcal{N}(\mathcal{S}_r)\backslash\mathcal{S}_r$ for any given round $r$ after initialization.
Each pair $\Bar{X}_j,\Bar{Z}_j$ correspond to a logical $X$ and $Z$ operator for a logical qudit, thus must satisfy the same commutation relations.
Due to the periodic evolution of the stabilizer group, i.e. the ISGs, the logical operators of the qudit Floquet code also evolves dynamically over the measurement rounds.

In each round $r$ (after initialization), logical operators can be defined as a product of Paulis on the vertices of a non-trivial loop (i.e. chain of edges that are not on the boundaries of any set of plaquettes).
On a genus $g$ lattice we can identify $2g$ pairs of non-trivial loops, where on each ``hole'' in the lattice we can define two pairs of logical operators $\Bar{X}_j,\Bar{Z}_j$ and $\Bar{X}_{j'},\Bar{Z}_{j'}$ corresponding to $X$ and $Z$ logical operators of two logical qudits on four non-trivial loops.
Two of these logical operators which commutes are on two non-trivial loops in the same direction, one of which we call a \textit{type-1 logical operator} and the other a \textit{type-2 logical operator}.
These operators are illustrated for the honeycomb qudit Floquet code in Fig.~\ref{fig:qudit_honeycomb_lattice}.
The two type-1 logical operators are illustrated as an operator with support on yellow edges (vertical loop) and an operator with support on pink edges (horizontal loop), which we identify as a representation of $\Bar{Z}_1$ and $\Bar{X}_2$ operator, respectively.
Meanwhile, two type-2 logical operators are illustrated as an operator with support on orange edges (vertical loop) and an operator with support on purple edges (horizontal loop), which we identify as a representation of $\Bar{Z}_2$ and $\Bar{X}_1$ operator, respectively.

A type-1 logical operator in round $r$ consists of Pauli operators on a non-contractible loop with non-identity Paulis on the round $r$ checks.
On the other hand, a type-2 logical operator in round $r$ consists of Pauli operators on a non-contractible loop with non-identity Paulis on the round $r+1$ checks.
To determine the form of each logical operator, consider two non-contractible loops in horizontal direction and vertical direction on a hole in the lattice $L_x,L_z\subset E$ (each is a set of edges in that loop).
On each loop we can define a type-1 logical operator and a type-2 logical operator.
A type-2 logical operator on $L_x$ of $\mathcal{S}_r$ has support on edges of round $r+1$ checks where the Paulis on each edge is the exponentiated Paulis of the current round's checks on the vertices in the loop.
We note that this approach can be thought of as a generalization of how logical operators are defined for the hyperbolic qubit Floquet code construction in~\cite{fahimniya2023hyperbolic}.
The difference here is how the non-identity Paulis are defined as we do not require the checks of the same color to correspond to the same Pauli observable, whereas all checks of the same color correspond to the same two-qubit Pauli observable in~\cite{fahimniya2023hyperbolic}. 
For example, a type-2 logical operator on $L_x$ in round $r=0\mod{3}$ has support on red edges and takes the form
\begin{equation}\label{eqn:type2_logical}
    \prod_{(v,u) \in L_x\cap E_r} [P_{g_v}^{\alpha_{(v,u)}}]_v [P_{g_u}^{\alpha_{(v,u)}}]_u \;,
\end{equation}
for some $\alpha_{(v,u)}\in[D]$.
A type-1 logical operator on $L_z$ of $\mathcal{S}_r$ is defined similarly, where it has support on edges of round $r$ checks where the Paulis on each edge is the exponentiated Paulis of the next round's checks on the vertices in the loop.
For example, a type-1 logical operator on $L_z$ in round $r=0\mod{3}$ has support on green edges and takes the form
\begin{equation}\label{eqn:type1_logical}
    \prod_{(v,u) \in L_z\cap E_g} [P_{r_v}^{\beta_{(v,u)}}]_v [P_{r_u}^{\beta_{(v,u)}}]_u \;,
\end{equation}
for some $\beta_{(v,u)}\in[D]$.

To show that these are indeed logical operators of round $r$ we need to show that they are operators in $\mathcal{N}(\mathcal{S}_r)\backslash\mathcal{S}_r$.
Suppose that checks $C_l$ on color $l$ edges is performed in round $r$.
First we note that since the logical operators are defined on a non-contractible loop, they cannot be generated by plaquette stabilizers of $\mathcal{S}_r$ as product of plaquettes with support on any edge $(v,u)\in E$ is identity due to condition 3 of Theorem~\ref{thm:general_qudit_floquet_code_conditions}.
A type-1 logical operator has support on edges of the color-$l$ checks but has the Paulis of the checks performed in the next round, hence it cannot be generated by $C_l$ as the Paulis of different colored checks on any vertex do not commute by condition 2 of Theorem~\ref{thm:general_qudit_floquet_code_conditions}.
A type-2 logical operator on the other hand, has support on edges of the checks performed in the next round, and hence also cannot be generated by $C_l$.
Thus we can conclude that the logical operators are not in $\mathcal{S}_r$.

To show that they are not in $\mathcal{N}(\mathcal{S}_r)$, the normalizer of $\mathcal{S}_r$, it is sufficient to show that they commute with the generators of $\mathcal{S}_r$, namely the plaquette operators and the checks of round $r$.
First we show that they commute with checks in round $r$.
For a type-1 logical operators, since they consist of Paulis of checks performed in round $r+1$ on edges of color $l$, it can be shown that they commute with color $l$ checks using condition 1 of Theorem~\ref{thm:general_qudit_floquet_code_conditions}.
A type-2 logical operators on the other hand, consist of Paulis of the color $l$ checks on edges of the next round checks, hence they must commute with the color $l$ checks.

Now we show that the logical operators commute with the plaquette operators.
Let us denote $l'$ as the color of next round's check and $l''$ as the color of previous round's check.
Consider a type-1 logical operator on loop $L$.
By condition 3 of Theorem~\ref{thm:general_qudit_floquet_code_conditions}, we can show that type-1 logical operators commute with the color-$l'$ plaquettes as these plaquettes have Paulis $P_{v_l}P_{v_{l''}}$ and $P_{u_l}P_{u_{l''}}$ on each edge $(v,u)\in L\cap E_l$.
To show that type-1 logical operators commute with color-$l''$ plaquettes, note that a color-$l''$ plaquette sharing qudits with a type-1 logical operator have Paulis $P_{v_l}P_{v_{l'}}$ and $P_{u_l}P_{u_{l'}}$ on each edge $(v,u)\in L\cap E_l$.
Since $c(P_{v_{l'}} , P_{v_l}) = -c(P_{u_{l'}} , P_{u_l})$ by condition 1, we can conclude that type-1 logical operators commute with color-$l''$ plaquettes.
Lastly, to show that type-1 logical operators commute with color-$l$ plaquettes, note that each of these plaquettes shares two qudits $v,v'$ with a type-1 logical operator and these two qudits are on two different color-$l$ edges, $(v,u),(v',u')\in L\cap E_l$.
Pauli operators of the plaquette stabilizer of these qudits are $P_{v_{l'}}P_{v_{l''}}$ and $P_{v_{l'}'}P_{v_{l''}'}$.
Thus we get commutation functions between Paulis of the plaquette and logical operators on qudits $v$ and $v'$ as
\begin{equation}
\begin{aligned}
    c(P_{v_{l'}}P_{v_{l''}} , P_{v_{l'}}^{\beta_{(v,u)}}) = \beta_{(v,u)} c(P_{v_{l''}} , P_{v_{l'}}) \\
    c(P_{v_{l'}'}P_{v_{l''}'} , P_{v_{l'}'}^{\beta_{(v',u')}}) = \beta_{(v',u')} c(P_{v_{l''}'} , P_{v_{l'}}) \;.
\end{aligned}
\end{equation}
Commutation between this plaquette stabilzier and the logical operator can then be guaranteed by setting 
\begin{equation}
    \beta_{(v',u')} = -\frac{\beta_{(v,u)} c(P_{v_{l''}} , P_{v_{l'}})}{c(P_{v_{l''}'} , P_{v_{l'}})} \;.
\end{equation}
Thus, we have shown that the type-1 logical operators commute with all plaquette stabilizers.

We can show that type-2 logical operators commute with the plaquette stabilizers similarly.
To show that it commutes with color-$l$ plaquettes and color-$l''$ plaquettes, we can use similar argument by using condition 3 and condition 1 of Theorem~\ref{thm:general_qudit_floquet_code_conditions}, respectively.
To show that it commutes with a color-$l'$ plaquette sharing qudits $v$ and $v'$ with the type-2 logical operator on edges $(v,u),(v',u')\in L\cap E_{l'}$, we can again look at the commutation function between Paulis of the plaquette and the logical operator at qudits $v$ and $v'$.
Then, setting 
\begin{equation}
    \alpha_{(v',u')} = -\frac{\alpha_{(v,u)} c(P_{v_{l''}} , P_{v_{l}})}{c(P_{v_{l''}'} , P_{v_{l}})}
\end{equation}
guarantees that the type-2 logical operators commute with this color-$l'$ plaquette.

Thus we have shown that both type-1 and type-2 logical operators in round $r$ defined on non-trivial loops on the lattice commute with all round $r$ checks and commute with all plaquette stabilizers while not in the stabilizer group $\mathcal{S}_r$ themselves, showing that they are in $\mathcal{N}(\mathcal{S}_r)\backslash\mathcal{S}$.
Lastly, we also note that for a type-1 and a type-2 logical operators on nontrivial loops $L_1$ (horizontal) and $L_2$ (vertical) which has the form of
\begin{equation}
\begin{gathered}
    Q_1 = \prod_{(v,u)\in L_1\cap E_{l'}} [P_{l_v}^{\alpha_{(v,u)}}]_v [P_{l_u}^{\alpha_{(v,u)}}]_u \\
    \text{and}\\
    Q_2 = \prod_{(v,u)\in L_2\cap E_{l}} [P_{l_v'}^{\alpha_{(v,u)}}]_v [P_{l_u'}^{\alpha_{(v,u)}}]_u
\end{gathered}
\end{equation}
we can identify them as $\Bar{X}_j$ and $\Bar{Z}_j$ logical operators.
This is done by identifying the qudit $w$ which is in both $L_1$ and $L_2$, then raise $Q_1$ to some power $a\in[D]$ such that their commutation function satisfy $c(Q_1^a,Q_2)=1$ so that we can define $\Bar{X}_j:=Q_1^a$ and $\Bar{Z}_j:=Q_2$.

\section{Error syndromes and decoding of qudit Floquet code}\label{sec:qudit_code_error_syndrome_decoding}

\begin{figure}
    \centering
    \includegraphics[width=1\columnwidth]{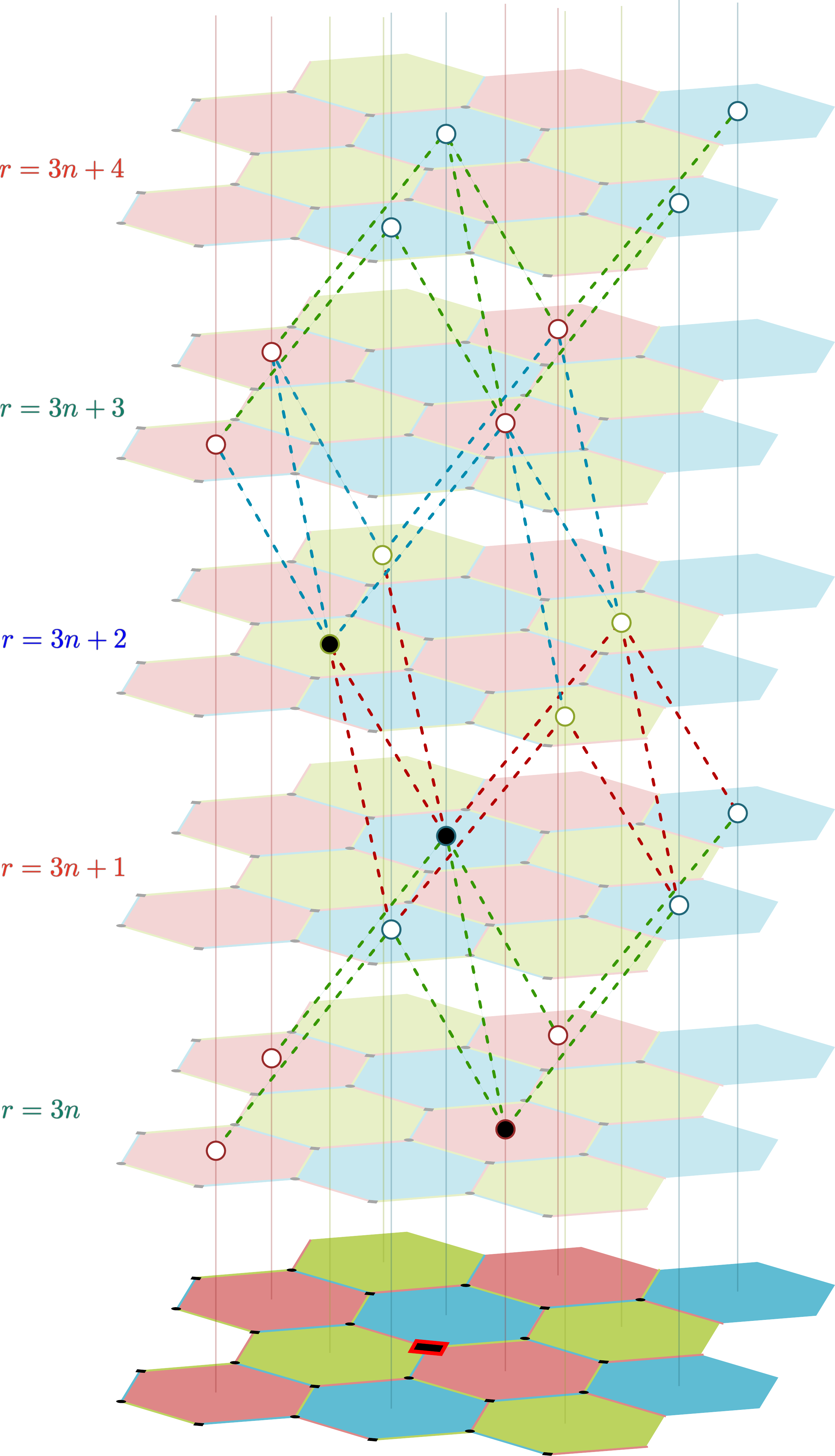}
    \caption{Space-time syndrome lattice over $5$ rounds.
    An error occuring on a qudit on the square vertex with red outline before round $r=3n$ (illustrated by square vertex with red outline) flips at least two out of three checks with support on that qudit in the next three rounds.
    The measurement outcomes of these three checks are used to infer the eigenvalue of the three plaquettes with support on the erroneous qudit.
    These potential error syndromes on the three plaquettes are represented by three black dots on the space-time lattice.
    Details on how error syndromes for the plaquette stabilizers are inferred is given in Section~\ref{sec:qudit_code_error_syndrome_decoding}.}
    \label{fig:syndrome_lattice}
\end{figure}

Here we illustrate how a single qudit error is detected and corrected by the qudit Floquet code.
A single qudit Pauli error $E\in\mathcal{P}_D$ on qudit on vertex $v\in V$ is detected by check measurement $[P_{v_l}]_v [P_{u_l}]_u$ when $E$ does not commute with $P_{v_l}$.
Suppose a single qudit error $E$ occurs on qudit $v$ right before round $r=0\mod{3}$ and no more errors over the next three rounds, then error $E$ is detected when
\begin{equation}
\begin{gathered}
    c(P_{v_g},E)\neq 0 \;,\mathrm{or}\\
    c(P_{v_r},E)\neq 0 \;,\mathrm{or}\\
    c(P_{v_b},E)\neq 0 \;.    
\end{gathered}
\end{equation}
Note that the Paulis $P_{v_g},P_{v_r},P_{v_b}$ of the green, red, and blue checks on qudit $v$ are distinct since they are pairwise non-commuting so that when an error $E$ occurs on a qudit it must flip at least two out of three checks with qudit $v$ on their support, i.e. gives a different measurement outcome compared to when it was last measured.
Each of the detection by a check will result in a change of the inferred eigenvalue of one or more plaquette stabilizers.

An error syndrome $s$ obtained in round $r$ corresponding to a plaquette stabilizer is represented by a vertex at the center of a plaquette which syndrome is inferred in that round.
Syndromes for the red plaquettes is obtained in round $r=0\mod{3}$, while syndrome for blue plaquettes are obtained in round $r=1\mod{3}$ and for green plaquettes, in round $r=2\mod{3}$.
These syndromes is used to construct a \textit{space-time syndrome lattice} where a syndrome $s$ is connected by three syndrome lattice edges to three syndromes of neighboring plaquettes inferred in the next round $r+1$ and to three other syndromes of neighboring plaquettes inferred in the previous round $r-1$.

In the illustration of Fig.~\ref{fig:syndrome_lattice}, a Pauli error $E$ occurs on a qudit $v$ illustrated by a square with red outline in between round $3n-1$ and $3n$.
In rounds $3n$, $3n+1$, and $3n+2$ green, red, and blue checks are measured, respectively.
In round $3n$, the eigenvalue of the red checks are inferred from measurement outcomes of the green checks in this round and the outcomes of the blue checks in the previous round.
If $E$ does not commute with the Pauli $P_{v_g}$ of the green check on vertex $v$, then the outcome of this check is not equal to the outcome of that same green check in round $3n-3$.
Thus if there is no other errors before round $3n$ then we can infer the eigenvalue of the red plaquette with support on $v$ which gives us an error syndrome $s_r\in[D]$ for that red plaquette stabilizer.
In round $3n+1$, eigenvalues of the blue plaquettes are inferred from red checks in this round and the green checks in the previous round.
A syndrome $s_b\in[D]$ for the blue plaquette with support on $v$ can be obtained similarly to the syndrome for red plaquettes in the last round.
The difference here being that when $c(E,P_{v_g}) = -c(E,P_{v_r})$, we can infer that no error occurs on that blue plaquette stabilizer since its Pauli operator on qudit $v$ is $P_{v_g}P_{v_r}$, which indicates that error $E$ commutes with this blue plaquette stabilizer.
Lastly in round $3n+2$, eigenvalues of the green plaquettes are inferred from the blue checks measured in this round and the red checks in last round.
Error syndrome for the green plaquette with support on $v$, if any, can be inferred similarly to the previous two rounds.

One can further consider a simplified error channel consists of an $X$-type error channel and a $Z$-type error channel affecting each qudit independently in each round (after check measurement) is one such channel that has been studied for qudit codes (see e.g.~\cite{anwar2014fast,watson2015fast}).
The $X$-type channel applies $X^i$ error to the codestate with probability $\frac{p}{d-1}$ and no error with probability $1-p$, whereas the $Z$-type channel applies $Z^j$ error with probability $\frac{p}{d-1}$ and no error with probability $1-p$, independent of the $X$-type channel (for $i,j\in\{1,\dots,d-1\}$).
The combination of both $X$-type and $Z$-type error channels can be expressed as an error channel $\mathcal{E}$ on a qudit density operator $\rho$ as
\begin{equation}
\begin{aligned}
    &\mathcal{E}(\rho) = (1-p)^2 \rho + \Big(\frac{p}{d-1}\Big)^2 \sum_{i,j=1}^{d-1} X^iZ^j \rho (Z^\dag)^j(X^\dag)^i \\
    &\; +\frac{(1-p)p}{d-1}\sum_{i'=1}^{d-1} X^{i'}\rho(X^\dag)^{i'} +\frac{(1-p)p}{d-1}\sum_{j'=1}^{d-1} Z^{j'}\rho(Z^\dag)^{j'} \,.
\end{aligned}
\end{equation}
Kraus operators of the simplified error channel $\mathcal{E}$ corresponds to elements of the (phaseless) qudit Pauli group $\mathcal{P}_D$.
Therefore, each Kraus operator correspond to a pattern in the syndrome lattice which can be used in a decoding scheme to correct errors.
Various decoders suitable for qudit codes such as the renormalization group (RG) decoders based on befief propagation~\cite{anwar2014fast,duclos2013kitaev,hutter2015improved,watson2015fast} or qudit Trellis decoder~\cite{sabo2024trellis} can be applied to the space-time syndrome lattice to decode multiple errors.

\section{$(\vcirc,\vsq)$-Qudit Floquet codes}
\label{sec:cirsq_qudit_floquet_code}

\begin{figure*}
    \centering
    \includegraphics[width=0.85\linewidth]{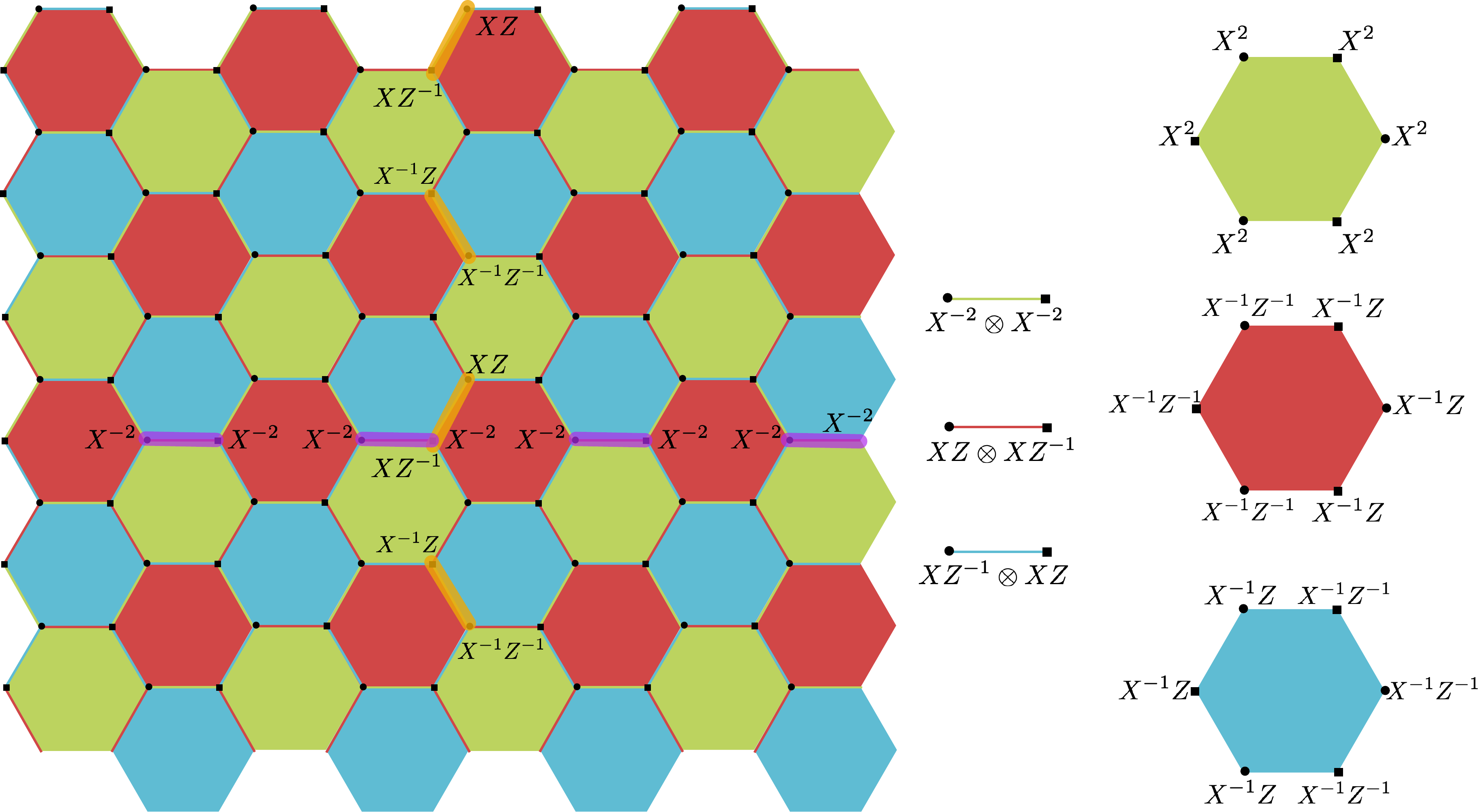}
    \caption{$(\vcirc,\vsq)$ Qudit Floquet code illustration.
    The left hand side figure is an illustration of a hexagonal lattice with periodic boundary condition for a 96 qudits Floquet honeycomb code.
    Three colors on the edges of the lattice represent different flavors of check measurements.
    The Pauli observable of the green, red, and blue checks are specified in the middle of the figure.
    Three colors on the hexagonal plaquettes represent three different flavors of plaquette operators of the code where Pauli operators on each $\vcirc$ and $\vsq$ for each color is given in the right hand side of the figure.
    Edges highlighted in orange and purple illustrate the type-1 and type-2 logical operators, respectively, in round $r=0\mod{3}$.
    As can be seen in the figure and verified for other rounds, the minimum weight of the logical operators in each round for the $96$ qudits Floquet code is $8$.
    Hence we can conclude that this code has a distance $d=8$ and an $\llbracket n,k,d \rrbracket$ code parameter of $\llbracket 96,2,8\rrbracket$.
    }
    \label{fig:qudit_honeycomb_lattice2}
\end{figure*}

Here we discuss a simple subfamily of the family of qudit Floquet codes defined in Section~\ref{sec:qudit_floquet_code} where checks of the same color correspond to the same weight-two measurement.
To define the $(\vcirc,\vsq)$-qudit Floquet codes, we first add an additional feature to the construction based on Theorem~\ref{thm:general_qudit_floquet_code_conditions} by assigning either a circle $\vcirc$ or a square $\vsq$ label to the vertices on a $\{p,3\}$-lattice with genus $g$ (again, for $p=6$ for $g=1$ and even $p\geq8$ for $g\geq2$) such that for any edge $(v,u)\in E$, exactly one of the vertices is of the type $\vcirc$ and the other is of the type $\vsq$.
Hence any $p$-gon has $p/2$ $\vcirc$ vertices and $p/2$ $\vsq$ vertices alternating on the vertices at its boundary.
An illustration for a $(\vcirc,\vsq)$-qudit Floquet codes on the honeycomb lattice, i.e. a $\{6,3\}$-tesselation on a genus $g=1$ surface is given in Fig.~\ref{fig:qudit_honeycomb_lattice2}.

We denote a Pauli $P_{a,b}$ assigned to a $\vcirc$-type vertex as $P_{a,b}^\vcirc$ and Pauli $P_{a'b'}$ assigned to a $\vsq$-type vertex as $P_{a',b'}^\vsq$.
A check is a measurement of tensor product of a $\vcirc$-type Pauli and a $\vsq$-type Pauli observable assigned to two vertices on an edge.
This gives us green, red, and blue checks of the form
\begin{equation}\label{eqn:checks}
\begin{gathered}
    [P_{g_1,g_2}^\vcirc]_v [P_{g_3,g_4}^\vsq]_u \,,\; (v,u)\in E_g \\
    [P_{r_1,r_2}^\vcirc]_v [P_{r_3,r_4}^\vsq]_u \,,\; (v,u)\in E_r \\
    [P_{b_1,b_2}^\vcirc]_v [P_{b_3,b_4}^\vsq]_u \,,\; (v,u)\in E_b \;,
\end{gathered}
\end{equation}
respectively for $g_1,g_2,g_3,g_4,r_1,r_2,r_3,r_4,b_1,b_2,b_3,b_4 \in [D]$.
For the rest of this section, when we talk about a check $[P_{l_1,l_2}]_v [P_{a_3,l_4}]_u$ we adopt the convention that check $P_{l_1,l_2}$ is a $\vcirc$-type check and $P_{l_3,l_4}$ is a $\vsq$-type check, i.e. $P_{l_1,l_2} := P_{l_1,l_2}^\vcirc$ and $P_{l_3,l_4} := P_{l_3,l_4}^\vsq$, when it is clear from the context.
We further require that these checks satisfy the following conditions:
\begin{enumerate}
    \item The two-qudit Paulis $P_{g_1,g_2} \otimes P_{g_3,g_4}$, $P_{r_1,r_2} \otimes P_{r_3,r_4}$, and $P_{b_1,b_2} \otimes P_{b_3,b_4}$ mutually commute.
    
    \item Paulis of checks with different color on the $\vcirc$ qudit and the $\vsq$ qudit do not commute, i.e. $c(P_{l_1,l_2}^\vcirc,P_{l_1',l_2'}^\vcirc) \neq 0$ and $c(P_{l_3,l_4}^\vsq,P_{l_3',l_4'}^\vsq) \neq 0$ for $l\neq l'$.
    
    \item The total number of $X$'s and $Z$'s in the $\vcirc$ Paulis and the total number of $X$'s and $Z$'s in the $\vsq$ Paulis over all the checks are both $0\mod{D}$, i.e. $g_j+r_j+b_j=0\mod{D}$ for $j\in\{1,2,3,4\}$.
\end{enumerate}

It can be verified that the checks of $(\vcirc,\vsq)$ qudit Floquet code satisfy the conditions of Theorem~\ref{thm:general_qudit_floquet_code_conditions}.
Since two-qudit Paulis $P_{g_1,g_2} \otimes P_{g_3,g_4}$, $P_{r_1,r_2} \otimes P_{r_3,r_4}$, and $P_{b_1,b_2} \otimes P_{b_3,b_4}$ mutually commute, they must satisfy condition 1 of Theorem~\ref{thm:general_qudit_floquet_code_conditions}.
The fact that Paulis of checks with different color on the $\vcirc$ qudit and the $\vsq$ qudit do not commute satisfy condition 2.
Lastly since the total number of both the $X$'s and the $Z$'s in both the $\vcirc$ and $\vsq$ Paulis of the checks are equal to $0\mod{D}$, they satisfy condition 3.
Thus the $(\vcirc,\vsq)$ qudit Floquet code has the same ISG dynamics, each of which encodes $2g$ logical qudits.

Now we provide an explicit construction of the $(\vcirc,\vsq)$ qudit Floquet code for arbitrary prime dimension $D\geq3$ with checks of the form
\begin{equation}\label{eqn:explicit_checks}
\begin{gathered}
    P_{g_1,g_2} \otimes P_{g_3,g_4} = X^{-2} \otimes X^{-2} \\
    P_{r_1,r_2} \otimes P_{r_3,r_4} = XZ \otimes XZ^{-1} \\
    P_{b_1,b_2} \otimes P_{b_3,b_4} = XZ^{-1} \otimes XZ \;.
\end{gathered}
\end{equation}
It can be verified that these checks satisfy the three requirements stated above for a $(\vcirc,\vsq)$ qudit Floquet code.
By using the qudit stabilizer update rules of Theorem~\ref{thm:qudit_stabilizer_update}, one can also obtain the green, red, and blue plaquette stabilizers with the following Paulis on its $\vcirc$ and $\vsq$ qudits:
\begin{equation}\label{eqn:explicit_plaquette_oper}
    \begin{array}{c|c|c}
        & \vcirc & \vsq \\
        \hline
        \mathrm{green} & X^2 & X^2 \\
        \mathrm{red} & X^{-1}Z^{-1} & X^{-1}Z \\
        \mathrm{blue} & X^{-1}Z & X^{-1}Z^{-1} \\ 
    \end{array} \;.
\end{equation}
A step-by-step update in the first 5 rounds (the initialization rounds) to obtain these plaquette stabilizers can be done similar to the proof outline of Theorem~\ref{thm:general_qudit_floquet_code_conditions} in Section~\ref{sec:ISG} (or in Appendix~\ref{app:general_qudit_floquet_code_conditions} for a detailed full proof).

\section{Discussion}

In this work we provide a construction of a large family of qudit Floquet codes which can be done on a torus and a large family of three-colorable hyperbolic lattices.
Our construction extends the existing proposals of dynamical quantum error-correcting codes~\cite{hastings2021dynamically,davydova2023floquet,fahimniya2023hyperbolic,gidney2021fault,haah2022boundaries,vuillot2021planar,paetznick2023performance,gidney2022benchmarking,higgott2023constructions,aasen2023measurement,zhang2023x,dua2024engineering,berthusen2023partial,fu2024error,mcewen2023relaxing,kesselring2024anyon,bombin2023unifying,townsend2023floquetifying,davydova2023quantum} which mostly focuses only on qubit systems, to qudit systems, as well as provide a construction of a family qudit ``Floquet'' codes where the code's measurement sequence is periodic which includes the qudit Floquet code construction on a torus proposed in~\cite{ellison2023floquet}.
In Theorem~\ref{thm:general_qudit_floquet_code_conditions} we provide a set of conditions on weight-2 Pauli measurements in the period-3 measurement schedule that we show to give a qudit Floquet codes in which the number of encoded logical qudits grows linearly with the genus (i.e. the number of ``holes'') of the lattice and an encoding rate that converges to $\frac{1}{2}$ asymptotically with the number of physical qudits and the size of faces on the lattice.
Additionally in Theorem~\ref{thm:qudit_stabilizer_update}, we show a set of rules to update a qudit stabilizer code after a Pauli measurement is performed on the code, analogous to qubit stabilizer update rules first proposed in~\cite{gottesman1998heisenberg} and used for stabilizer update of qubit dynamical codes (e.g. in~\cite{hastings2021dynamically,fu2024error}).

To further emphasize the utility of the qudit Floquet codes, one could also show a threshold behavior of the code by analytical or numerical means, e.g. for the simple qudit Pauli error channel.
While existing results on qubit Floquet codes such as~\cite{gidney2021fault,davydova2023floquet} uses maximum weight perfect matching (MWPM) decoder on the syndrome lattice of a qubit Floquet code to show a threshold behavior, this decoder is unsuitable for qudit codes due to the non-binary nature of the error syndromes.
For the qudit Floquet code, it would be interesting to adopt a decoder used for static qudit codes such as the renormalization group (RG) belief propagation decoders~\cite{duclos2013kitaev,anwar2014fast,watson2015fast,hutter2015improved} or the recently proposed qudit Trellis decoding~\cite{sabo2024trellis} to analyze any threshold behavior for the qudit Floquet code.
In doing this, one needs to generalize these qudit decoders to the qudit Floquet code error syndrome lattice discussed in Section~\ref{sec:qudit_code_error_syndrome_decoding}.

In a future work, one may also consider a ``planar'' qudit Floquet code on a lattice with boundaries as it has been done for qubit Floquet code in~\cite{vuillot2021planar,haah2022boundaries}, one may find interesting dynamics of the code from different measurement schedules due the presence of boundaries.
A curious phenomenon is shown in~\cite{vuillot2021planar} where the planar qubit Floquet code distance is constant with respect to the number of physical qubits, implying that the code does \textit{not} possess a fault-tolerant property.
This however, can be overcame through careful choice of measurement schedule and considerations of the boundaries~\cite{haah2022boundaries}.
It is interesting to see how different measurement schedules give rise to different planar qudit Floquet codes and how one should construct a measurement schedule such that the qudit ISGs encode logical qudits that grows with the number of physical qudits.

\section*{Acknowledgements}

AT is supported by CQT PhD Scholarship and Google PhD Fellowship program. KB acknowledges the support from A*STAR C230917003. This work is supported by the Singapore Ministry of Education Tier 1 Grants RG77/22 and RT4/23, the Agency for Science, Technology and Research (A*STAR) under its QEP2.0 programme (NRF2021-QEP2-02-P06), the CQT Bridging Grant, and FQxI under grant no ~FQXi-RFP-IPW-1903 (``Are quantum agents more energetically efficient at making predictions?") from the Foundational Questions Institute and Fetzer Franklin Fund (a donor-advised fund of Silicon Valley Community Foundation). Any opinions, findings and conclusions or recommendations expressed in this material are those of the author(s) and do not reflect the views of National Research Foundation or the Ministry of Education, Singapore.

\bibliographystyle{unsrt}
\bibliography{references}

\begin{appendix}

\onecolumngrid

\section{Proof of Theorem~\ref{thm:qudit_stabilizer_update}}\label{app:qudit_stabilizer_update_proof}

For case 1, the stabilizer code $\mathcal{C}(\mathcal{S})$ clearly left unchanged (i.e. $\mathcal{C}(\mathcal{S})=\mathcal{C}(\mathcal{S}')$) since all $|\psi\>\in\mathcal{C}(\mathcal{S})$ are stabilized by $\omega^a P$.
Now to show that the outcome $o$ is deterministic, note that there can only be one of $\omega^a P$ in $\mathcal{S}$ for all $a\in[D]$.
Since if $|\psi\>$ is stabilized by both $\omega^aP$ and $\omega^bP$ for $a\neq b$ (recall that all (in-)equalities and arithmetics are in $\mod{D}$), then $P|\psi\> = \omega^{-a}|\psi\>$ and $P|\psi\> = \omega^{-b}|\psi\>$ giving a contradiction $a=b$.
Thus if $\omega^a P\in\mathcal{S}$ then measurement of observable $\omega^a P$ on some state $|\psi\>\in\mathcal{C}(\mathcal{S})$ gives an outcome $o=a$ deterministically.

For case 2 and 3, we need to show 
that for any outcome $o$ upon measurement of $P$, state $|\psi'\>$ is in $\mathcal{C}(\mathcal{S}')$ if and only if it is a post-measurement state of some $|\psi\>\in\mathcal{C}(\mathcal{S})$.
In case 2, since all $|\psi'\>\in\mathcal{C}(\mathcal{S}')$ is a $+1$ eigenstate of all $g_1,\dots,g_m$, then it must be in $\mathcal{C}(\mathcal{S}')$.
Because it is also a $+1$ eigenstate of $\omega^{-o}P$, i.e. $\omega^{-o}P|\psi'\>=|\psi'\>$, thus $P|\psi'\>=\omega^o|\psi'\>$.
For projector $\Pi_o^{(P)}$ onto the $\omega^o$ eigenspace of $P$ (see eqn.~\eqref{eqn:qudit_pauli_eigenspace_projector}), we can see that $|\psi'\>$  is a post-measurement state of $|\psi'\>\in\mathcal{C}(\mathcal{S})$ from $\Pi_o^{(P)}|\psi'\>=|\psi'\>$.
Now consider a $|\psi\>\in\mathcal{C}(\mathcal{S})$.
Upon obtaining outcome $o$ from measurement of $P$, state $|\psi\>$ is projected onto the $\omega^o$ eigenspace of $P$ by projector $\Pi_o^{(P)}$, resulting in post-measurement state $|\psi'\>=\Pi_o^{(P)}|\psi\>$.
Since $P|\psi'\> = \omega^o|\psi'\>$, then $|\psi'\>$ is a $+1$ eigenstate of $\omega^{-o}P$.
As generators $g_1,\dots,g_m$ commute with $P$ and $g_j|\psi\>=|\psi\>$, therefore $|\psi'\>$ is also a $+1$ eigenstate of $g_j$ since $g_j|\psi'\> = g_j\Pi_o^{(P)}|\psi\> = \Pi_o^{(P)}g_j|\psi\> = \Pi_o^{(P)}|\psi\>=|\psi'\>$.
Thus, we complete the proof of case 2.

Now in case 3 first note that $g_1$ and $W'$ generates stabilizers in $w$ ,thus along with the generators $g_{l+1},\dots,g_m$ that commute with $P$ they generate the same stabilizer group $\mathcal{S}$.
Moreover the only element of this new set of generators $G'$ containing $g_1$, $W'$, and $g_{l+1},\dots,g_m$, its only element that does not commute with $P$ is $g_1$.
Also, all Paulis in $W'$ commute with $P$ (and with each other since $g_1,\dots,g_l$ mutually commute) and with $g_{l+1},\dots,g_m$.
First note that since all $|\psi\>\in\mathcal{C}(\mathcal{S})$ must be a $+1$ eigenstate of $g_1$, the probability of outcome $o$ from measuring $P$ is $\frac{1}{D}$ for any outcome $o$.
This is due to
\begin{equation}
\begin{aligned}
    \<\psi|P|\psi\> = \<\psi|g_1 P|\psi\> = \omega^{b_1} \<\psi|Pg_1|\psi\> = \omega^{b_1} \<\psi|P|\psi\> \;,
\end{aligned}
\end{equation}    
which implies that the expectation of measurement $P$ on state $|\psi\>$ is $\<\psi|P|\psi\> = 0$ (recall that $P = \sum_o \omega^o \Pi_o^{(P)}$ from eqn.~\eqref{eqn:qudit_pauli_eigenspace_projector}).
Thus, we know that for any outcome $o$, each $|\psi\>\in\mathcal{C}(\mathcal{S})$ is projected to a post-measurement state $|\psi_o'\>$ by projector $\Pi_o^{(P)}$ with non-zero probability.
Explicitly we have $|\psi_o'\> = \Pi_o^{(P)}|\psi\>$ (up to normalization), which implies that it is a $+1$ eigenstate of generators in $W'$ and $g_{l+1},\dots,g_m$ which commute with $P$ because $g|\psi_o'\> = g\Pi_o^{(P)}|\psi\> = \Pi_o^{(P)}g|\psi\> = \Pi_o^{(P)}|\psi\> = |\psi_o'\>$ for any $g$ in $W'$ or $g_{l+1},\dots,g_m$.
Since $|\psi_o'\>$ is the $\omega^o$ eigenstate of $P$ and a $+1$ eigenstate of generators in $W'$ and $g_{l+1},\dots,g_m$, it must be in $\mathcal{C}(\mathcal{S}')$.

Now we show that for any $|\psi'\>\in\mathcal{C}(\mathcal{S}')$ there exists $|\psi\>\in\mathcal{C}(\mathcal{S})$ such that $|\psi'\> \propto \Pi_o^{(P)}|\psi\> \neq 0$, i.e. a non-zero probability of obtaining outcome $o$ from measurement of $P$ on state $|\psi\>$.
Assume that stabilizer code $\mathcal{C}(\mathcal{S})$ encodes $k$ logical qudits, then consider basis $\{|\Bar{x}\> : x\in[D]^k\}$ for codespace $\mathcal{C}(\mathcal{S})$ where state $|\Bar{x}\>$ is stabilized by a subgroup generated by logical operators $\omega^{-x_1}\Bar{Z}_1,\dots,\omega^{-x_k}\Bar{Z}_k$ of $\mathcal{S}$ and $g_1,\dots,g_m$.
The logical operators can be transformed as $\Bar{Z}_i \mapsto \Bar{Z}_ig_1^{s_i}$ so that $\Bar{Z}_ig_1^{r_i}$ commute with $P$.
The new logical operator $\Bar{Z}_ig_1^{r_i}$ is merely a different representation of logical operator $\Bar{Z}_i$ as it is still a normalizer of $\mathcal{S}$ (so $\Bar{Z}_ig_1^{r_i}|\Bar{x}\>=|\Bar{x}\>$). 
Moreover they are also logical operators of the updated stabilizer $\mathcal{S}'$ as each $\Bar{Z}_ig_1^{r_i}$ commutes with all elements of $\mathcal{S}'$.
Basis $\{|\Bar{x}'\> : x\in[D]^k\}$ of the new code $\mathcal{C}(\mathcal{S}')$ is made out of state $|\Bar{x}'\>$ stabilized by a subgroup generated by the logical operators $\omega^{-x_1}\Bar{Z}_1g_1^{r_1},\dots,\omega^{-x_k}\Bar{Z}_kg_1^{r_k}$ and $\mathcal{S}'$.
Therefore the updated stabilizer code $\mathcal{C}(\mathcal{S}')$ also encodes $k$ logical qudits with updated basis $|\Bar{x}'\> \propto \Pi_o^{(P)}|\Bar{x}\>$ as the probability of outcome $o$ from measuring $P$ on $|\Bar{x}\>$ is non-zero.

\section{Proof of Theorem~\ref{thm:general_qudit_floquet_code_conditions}}\label{app:general_qudit_floquet_code_conditions}

We start with stabilizer that contains only the identity $\mathcal{S}_{-1}=\{I^{\otimes n}\}$.
In round $r=0$ we measure the green checks $C_g$ and obtain round-0 ISG $\mathcal{S}_0=\<C_g\>$ by invoking rule 2 of Theorem~\ref{thm:qudit_stabilizer_update}.

In round $r=1$ we measure the red checks $C_r$.
Since the degree of each vertex is three where the three edges consists of a green, a red, and a blue edge, each red check does not commute with two green checks that it shares its two vertices with by condition~\eqref{eqn:checks_noncommute}.
By measuring the $\frac{p}{2}$ red checks surrounding each blue plaquette one by one in clock-wise direction as illustrated in top diagram of Fig.~\ref{fig:plaquette_update}, we invoke update rule 3 of Theorem~\ref{thm:qudit_stabilizer_update} $\frac{p}{2}-1$ times, and invoke rule 2 one time.
For a given blue plaquette, a red check is non-commuting with two green checks that it shares its two vertices with but commutes with the product of these two green checks.
Hence rule 3 removes the two green checks from the stabilizer generator and adds the product of these two checks.

By similar reasoning, after measuring the second red check, we remove one of the $\frac{p}{2}$ green checks, remove the product of the last two green checks, and adds the product between three green checks to the stabilizer generator.
Doing this for $\frac{p}{2}-1$ red checks then gives us a weight-$p$ operator with support on the blue plaquette obtained by taking the products of $\frac{p}{2}$ red checks surrounding it.
The last red check commutes with this blue plaquette by the second line of condition 1 of Theorem~\ref{thm:general_qudit_floquet_code_conditions} and hence we apply update rule 2 and just add it to the list of generators.
By doing similar update on each blue plaquette, we obatin the round-1 ISG $\mathcal{S}_1 = \<A_b',C_r\>$ where $A_b'$ is a set consisting of weight-$p$ operators on each blue plaquette obtained by taking the products of red checks surrounding it.
We call the weight-$p$ Paulis in $A_b'$ the \textit{unformed} blue plaquette stabilizers (or, simply ``blue plaquette'' if clear from the context), as they are going to be updated in the next round.

In round $r=2$ we measure the blue checks $C_b$.
Note that each blue check is sharing one of its vertices with a red check and a blue plaquette stabilizer and sharing its other vertex with a different pair of red check and blue plaquette (see bottom diagram in Fig.~\ref{fig:plaquette_update}.
By condition 2 in Theorem~\ref{thm:general_qudit_floquet_code_conditions}, this blue check does not commute with each of the four stabilizers.
How the stabilizer update is performed in this round can be analyzed by considering a green plaquette and its three neighboring blue plaquettes from $A_b'$, then using update rules in Theorem~\ref{thm:qudit_stabilizer_update} as we measure $\frac{p}{2}$ red checks surrounding that green plaquette one by one in a clock-wise direction.
The first Blue check measured does not commute with two unformed blue plaquette stabilizers as well as two red checks.
Thus they will invoke update rule 3 which gives us an updated blue plaquettes and a weight-4 Pauli.
The weight-4 Pauli (illustrated as green edges in the second figure of the bottom diagram of Fig.~\ref{fig:plaquette_update}) is simply the product of two red checks, but an updated blue plaquette stabilizer is a product between an unformed blue plaquette and a red check on its boundary.
Suppose that this blue check is on the $(v,u)$ edge of the graph, so it takes the form $M_{(v,u)}=[P_{v_b}]_v [P_{u_b}]_u$.
On the other hand, the Pauli operator of two red checks and the two (unformed) blue plaquettes on vertices $v$ and $u$ are $P_{v_r}$ and $P_{u_r}$ and $P_{v_g}$ and $P_{u_g}$, respectively.
The weight-4 Pauli commutes with the blue check since by the last line of condition 1 of Theorem~\ref{thm:general_qudit_floquet_code_conditions},
\begin{equation}
\begin{aligned}
    &c(P_{v_b}\otimes P_{u_b} , P_{v_r}\otimes P_{u_r}) \\
    &= c(P_{v_b} , P_{v_r}) + c(P_{u_b},P_{u_r}) = 0 \;.
\end{aligned}
\end{equation}
Also by the third condition stated in Theorem~\ref{thm:general_qudit_floquet_code_conditions}, we have
\begin{equation}
\begin{aligned}
    & P_{v_b}P_{v_r}P_{v_g} \otimes P_{u_b}P_{u_r}P_{u_g} = I\otimes I \\
    &\Rightarrow P_{v_r}P_{v_g} = \omega^a P_{v_b}^{-1} \;\textup{and}\; P_{u_r}P_{u_g} = \omega^{-a} P_{u_b}^{-1}
\end{aligned}
\end{equation}
for some $a\in[D]$.
Note that one of the updated blue plaquette has a Pauli $P_{v_r}P_{v_g}$ on vertex $v$ and the other has a Pauli $P_{u_r}P_{u_g}$ on vertex $u$, hence each of the updated blue plaquettes commute with the blue check.
These two updated blue plaquettes are illustrated as hexagons with dark blue edges in the second figure of the bottom diagram of Fig.~\ref{fig:plaquette_update}, where the top blue plaquette is the product between the top unformed blue plaquette (top light blue in step 1) and the top red check.

The update from measuring the next blue check we measure is similar.
This blue check does not commute with a red check and an unformed blue plaquette sharing one of its vertices with, and does not commute with the weight-4 operator (green edge in the figure) obtained from the last update.
By invoking rule 3, the red check and the weight-4 operator are removed and a weight-6 operator obtained by taking the product of three red checks is then added to the set of stabilizer generators.
Also by rule 3, the unformed blue plaquette is updated by taking its product with the red check as in the update for the previous blue check measurement.
Updates for the remaining blue checks are done similarly, except for the last one where by using condition 3 of Theorem~\ref{thm:general_qudit_floquet_code_conditions} one can show that it commutes with updated blue plaquettes.
Also by the last equality of condition 1 in Theorem~\ref{thm:general_qudit_floquet_code_conditions}, this last blue check commutes with the weight-$p$ stabilizer on the green plaquette obtained by taking the products of all of its surrounding red checks.
By repeating this for each green plaquettes and three blue plaquettes surrounding it, we obtain the set of unformed green plaquette stabilizers $A_g'$ and initialized blue plaquette stabilizers $A_b$ where each plaquette in $A_b$ is the product between an unformed blue plaquette in $A_b'$ with its surrounding three red checks.
Thus at the end of this round we get the round-2 ISG $\mathcal{S}_2=\<A_g',A_b,C_b\>$.

Now in round $r=3$ we measure the green checks $C_g$.
The update is similar to the update in round $r=2$ where we use update rules in Theorem~\ref{thm:qudit_stabilizer_update} as we successively measure the $\frac{p}{2}$ green checks surrounding each red plaquette in a clock-wise order.
Note that each green check does not commute with two unformed green plaquette operators and two blue checks.
However, it commutes with the blue plaquette that it is sharing two vertices with.
Consider a red check on edge $(v,u)\in E_r$, so that it takes the form of $M_{(v,u)} = [P_{v_r}]_v [P_{u_r}]_r$.
The blue plaquette with support on qudits $v$ and $u$ has Paulis $P_{v_r}P_{v_g}$ and $P_{u_r}P_{u_g}$ on those qudits, respectively.
By line 2 of condition 1 in Theorem~\ref{thm:general_qudit_floquet_code_conditions} we have
\begin{equation}\label{eqn:checks_plaquettes_commute}
\begin{aligned}
    &c(P_{v_r}\otimes P_{u_r} , P_{v_r}P_{v_g} \otimes P_{u_r}P_{u_g})\\ 
    &= c(P_{v_r} , P_{v_r}P_{v_g}) + c(P_{u_r} , P_{u_r}P_{u_g}) \\
    &= c(P_{v_r} , P_{v_g}) + c(P_{u_r} , P_{u_g}) = 0 \;,
\end{aligned}
\end{equation}
which confirms that the red checks $C_r$ commute with blue plaquettes $A_b$.
The first $\frac{p}{2}-1$ green check measurements will give us an unformed red plaquette operator which is a product of blue checks surrounding it.
Then for the last green check we can show that it commutes with this unformed red plaquette by the first line of condition 1 in Theorem~\ref{thm:general_qudit_floquet_code_conditions}.
The first $\frac{p}{2}-1$ green check measurements will also update the unformed the green plaquette operators that are neighbors of the red plaquette.
Each updated green plaquette is simply the product between the unformed green plaquette with the blue check which in the boundary of both the green and red plaquette.
Repeating this update for all red plaquettes gives us the unformed red plaquettes $A_r'$ and update the green plaquettes.
Hence by the end of this round we get the round-3 ISG $\mathcal{S}_3 = \<A_r',A_g,A_b,C_g\>$.

Round $r=4$ is the last initialization round of the qudit Floquet code, in which we measure the red checks.
Each red check commutes with the green and blue plaquettes as can be verified by using line two and three of condition 1 in Theorem~\ref{thm:general_qudit_floquet_code_conditions} using similar argument as eqn.~\ref{eqn:checks_plaquettes_commute}.
However, each red check does not commute with two unformed red plaquettes and two green checks.
Similar to update in round $r=3$, we measure the red checks surrounding each blue plaquette successively in the clock-wise direction which updates the unformed red plaquettes by multiplying them with the green checks in its surrounding.
This is done similarly to the update for unformed green plaquettes in round $r=3$ and the update for unformed blue plaquettes in round $r=2$, i.e. by using rule 3 and 2 of Theorem~\ref{thm:qudit_stabilizer_update}, then the fact that the updated red plaquettes commute with the red checks can be shown using condition 3 of Theorem~\ref{thm:general_qudit_floquet_code_conditions}.
Finally by denoting $A_r$ as the updated red plaquettes, we obtain the ISG of round 4 $\mathcal{S}_4 = \<A_r,A_g,A_b,C_g\>$.

For round $r>4$, all checks commute with all plaquettes $A_g,A_r,A_b$ as established in the update of round 2, 3, and 4 above.
Hence the plaquette operators $A_g,A_r,A_b$ remains in the set of generators for ISG $\mathcal{S}_r$ for all $r>4$, whereas only exactly one of the set of checks $C_g,C_r,C_b$ is in the set of generators in period of three.

For each ISG $\mathcal{S}_r$ after initialization, 
each ISG encodes $k=2g$ logical qudits, where $g$ is the genus of the tesselated manifold\footnote{Genus $g$ has a relation $g=\frac{n}{16}+1$ to the number of physical qudits $n$. Give $k=2g$ encoded logical qudits, this implies that the encoding rate of $\lim_{n\rightarrow\infty}\frac{k}{n} = \frac{1}{8}+\frac{2}{n} = \frac{1}{8}$ qudits for any $\{8,3\}$-qudit hyperbolic Floquet code, as argued in~\cite{fahimniya2023hyperbolic}).}.
To see that there are $k=2g$ encoded logical qudits for a qudit Floquet code with $\{p,3\}$ tesselation on genus $g$ manifold, we can identify the number of independent generators for an ISG $\mathcal{S}_r$ and count the number of plaquettes and check generators\footnote{A similar explicit argument is made in~\cite{albuquerque2009topological} for a generalization of Kitaev's Toric code~\cite{kitaev2003fault} to higher genus manifold with arbitrary tesselation. In their case however, the stabilizer generators are defined on the plaquettes of lattice and dual lattice, as opposed to the generators of the qudit Floquet code ISGs which are defined on plaquettes and checks (color-$l$ edges).}.
Let us denote the number of plaquettes by $n_p$ and the number of edges by $n_e$.
The set of plaquette stabilizers are not independent since the product of all the plaquette operators $A_g,A_r,A_b$ is the identity $I^{\otimes n}$.
This can be seen by taking an arbitrary edge $(v,u)$ and take the product between the Paulis of the plaquettes with support on qudits $v$ and $u$ (note that there are exactly three plaquettes of each color with support on each vertex):
\begin{equation}
\begin{aligned}
    &\big(P_{v_r}P_{v_b} \otimes P_{u_r}P_{u_b}\big) \big(P_{v_b}P_{v_g} \otimes P_{u_b}P_{u_g}\big) \big(P_{v_g}P_{v_r} \otimes P_{u_g}P_{u_r}\big) \\
    &= \big( P_{v_g} P_{v_r} P_{v_b} \otimes P_{u_g} P_{u_r} P_{u_b} \big)^2 \\
    &= I\otimes I
\end{aligned}
\end{equation}
where the first equality and the third equality are by condition 1 and 3 of Theorem~\ref{thm:general_qudit_floquet_code_conditions}, respectively.
Thus the number of independent plaquette stabilizer generators is $n_p-1$.
The number of independent checks in the generator set of ISG $\mathcal{S}_r$ after initialization where check $C_l$ is measured is $|E_l|-1$.
This is due to the product between the color-$l$ checks and color-$l$ plaquettes for each $l\in\{g,r,b\}$ is equal to $I^{\otimes n}$, as can be seen by
\begin{equation}
\begin{aligned}
    \Big(P_{v_r}P_{v_b} \otimes P_{u_r}P_{u_b}\Big) \, P_{v_g}\otimes P_{u_g} = I\otimes I ,\;\forall(v,u)\in E_g \\
    \Big(P_{v_b}P_{v_g} \otimes P_{u_b}P_{u_g}\Big) \, P_{v_r}\otimes P_{u_r} = I\otimes I ,\;\forall(v,u)\in E_r \\
    \Big(P_{v_g}P_{v_r} \otimes P_{u_g}P_{u_r}\Big) \, P_{v_b}\otimes P_{u_b} = I\otimes I ,\;\forall(v,u)\in E_b 
\end{aligned}
\end{equation}
again by condition 1 and 3 of Theorem~\ref{thm:general_qudit_floquet_code_conditions}.

The number of logical qudits $k$ can be obtained by $D^k = \frac{D^n}{D^{n_p-1 + |E_l|-1}}$, or equivalently $k = n- (n_p + |E_l|-2)$.
Since any two checks of color $l$ have support on two different pairs of qudits, therefore $|E_l|=\frac{n}{2}$.
For honeycomb lattice (i.e. $\{6,3\}$-tesselation of a torus), we have $n_p=\frac{n}{2}$ plaquettes and hence each ISG encodes $k = n- (n-2) = 2$ qudits.
For a $\{p,q\}$-tesselation on a genus-$g$ manifold such that $(p-2)(q-2)>4$ it is shown in~\cite{albuquerque2009topological} that the number of plaquettes must satisfy
\begin{equation}\label{eqn:num_plaquettes_equality}
    4(g-1) = n_p\Big(p-2-\frac{2p}{q}\Big) 
\end{equation}
and the number of edges is $n_e = n_p\frac{p}{2}$ and the number of vertices $n=n_p\frac{p}{q}$.
For a three-colorable $\{p,3\}$-tesselation, the number of edges with color $l\in\{g,r,b\}$ is $|E_l| = n_p\frac{p}{6}$ and $n=n_p\frac{p}{3}$.
Thus the number of encoded qudits is
\begin{equation}\label{eqn:relation_logical_qudits_plaquettes_hyeprbolic}
\begin{aligned}
    k &= n - (n_p+|E_l|-2) \\
    &= \frac{n_pp}{6} - n_p + 2 \;,
\end{aligned}
\end{equation}
then by substituting in $q=3$ to eqn.~\eqref{eqn:num_plaquettes_equality} and rearranging, we get
\begin{equation}
\begin{aligned}
    4(g-1) &= n_p\Big(p-2-\frac{2p}{3}\Big) \\
    2g-2 &= \frac{n_p}{2}\Big( \frac{p}{3}-2 \Big) \\
    2g &= \frac{n_pp}{6} - n_p + 2 \;.
\end{aligned}
\end{equation}
Hence we have shown that $k=2g$, which implies that the number of encoded logical qudits $k$ in a three-colorable $\{p,3\}$ lattice for even $p\geq8$ on genus-$g$ manifold is $2g$.

\end{appendix}

\end{document}